\documentclass[runningheads]{llncs}
\usepackage{amsmath,amssymb}
\usepackage[utf8]{inputenc}
\usepackage{graphicx}
\usepackage[mathscr]{euscript}
\usepackage[pdftex,dvipsnames]{xcolor}
\usepackage{algorithm}
\usepackage{algpseudocode}
\usepackage{algpascal}
\usepackage{cite}
\usepackage{hyperref}
\usepackage[autostyle]{csquotes}

\usepackage{subfig}
\usepackage{float}

\newcommand{\lift}[1]{\textit{Lift$(#1)$}}
\newcommand{\liftk}[1]{\textit{$\overline{Lift}(#1)$}}
\newcommand{\uvec}[1]{\boldsymbol{\hat{\textbf{#1}}}}

\newcommand{\ball}[3]{\ensuremath{B(#1(#2),#3)}}
\newcommand{\dist}[3]{\ensuremath{dist_{#1}(#2, #3)}}
\newtheorem{observation}{Observation}
\newcommand{\NP}{\textbf{NP}}
\newcommand{\coNP}{\textbf{co-NP}}

\usepackage[colorinlistoftodos,prependcaption]{todonotes}

\newcommand{\decomplgo}{%
\alglanguage{pascal}
\begin{algorithm}[H]
\caption{Path decomposition algorithm}
\hspace*{\algorithmicindent} \textbf{Input:} Tree $T$ \\
\hspace*{\algorithmicindent} \textbf{Output:} List $\mathcal{P}$ of disjoint paths in $T$
\begin{algorithmic}[1]
\State $\mathcal{P} = \emptyset$
\State Forest $F = \{T\}$
\While {\text{there are edges in $F$}}
	\State Select root-to-leaf part $P'$ from tree $T'$
    \State ~~~~Delete edges of $P'$ from $T'$
    \State ~~~~Insert $P'$ in the beginning of $\mathcal{P}$
    \State ~~~~Delete $T'$ from $F$
    \State ~~~~Add trees from the $T' \setminus P$ to $F$
\EndWhile
\State Return $\mathcal{P}$
\end{algorithmic}
\end{algorithm}

}

\title{Morphing tree drawings in a small 3D grid}
\titlerunning{Morphing tree drawings in a small 3D grid}

\begin{document}

\title{Morphing tree drawings in a small 3D grid}
%
%
\author{Elena  Arseneva 
\and
Rahul Gangopadhyay 
\and
Aleksandra  Istomina  
}
\authorrunning{E. Arseneva et al.}
%
\institute{Saint-Petersburg University, Saint Petersburg, Russia \\ 
\email{e.arseneva@spbu.ru}, \email{rahulg@iiitd.ac.in}, \email{st062510@student.spbu.ru}
}
\maketitle              

\vspace{-2em}
\begin{abstract}
We study crossing-free grid morphs for planar tree drawings using 3D. A morph 
consists of morphing steps, 
where vertices move simultaneously along straight-line trajectories at  constant speeds. 
A crossing-free morph is known between two drawings of an $n$-vertex planar graph $G$ with $\mathcal{O}(n)$ morphing steps and  using the third dimension it can be reduced to $\mathcal{O}(\log n)$ for an $n$-vertex tree [Arseneva et al.\ 2019]. 
However, these morphs do not bound one practical parameter, the resolution. 
Can the number of steps be reduced substantially  by using the third dimension while keeping the resolution bounded throughout the morph?
We answer this question in an affirmative and present a 3D non-crossing morph between two planar grid drawings of an $n$-vertex tree in $\mathcal{O}(\sqrt{n} \log n)$ morphing steps. 
Each intermediate drawing 
lies in a $3D$ grid of polynomial volume.  

\small{}
\end{abstract}

 \vspace{-1.7em}

\section{Introduction}

Given an $n$-vertex graph $G$,  
a \emph{morph} between two drawings (i.e., embeddings in $\mathbb{R}^d$) of $G$ is a continuous transformation from one drawing to the other through a family of intermediate drawings.  
One is interested in well-behaved morphs, i.e., those that preserve essential properties of the drawing at any moment. 
Usually, this property is that the drawing is \emph{crossing-free}; 
such morphs are called \emph{crossing-free} morphs. 
This concept finds applications in multiple domains: animation, modeling, and computer graphics, etc.~\cite{gs2001controllable, gs-gifm-pm, gs2003intrinsic, fg-hmti-99}.
A drawing of 
$G$ is a \emph{straight-line drawing} if it maps each vertex of $G$ to a point in $\mathbb{R}^d$ and each edge of $G$ to the line segment whose endpoints correspond to  the endpoints of this edge. In this work, 
we focus on the case of drawings in the Euclidean plane ($d = 2$) and $3D$ drawings ($d = 3$); a non-crossing drawing of a graph in $\mathbb{R}^2$ is called \emph{planar}. 

The study of crossing-free graph morphs began 
in 1944 from a proof that such morphs for maximal planar graphs exist between two topologically equivalent drawings~\cite{c-dplc-44}; a matching algorithmic result appeared in 1983~\cite{t-dpg-83}. 
Recently crossing-free morphs of straight-line drawings are considered, where the vertex trajectories are simple. 
Of particular interest is morph transforming one straight-line drawing $\Gamma$ of a graph $G$ to another such drawing $\Gamma'$
through a sequence $\langle \Gamma = \Gamma_1,\Gamma_2,\ldots,\Gamma_k = \Gamma' \rangle$ of straight-line drawings. \emph{Linear morph} between $\langle \Gamma_i,\Gamma_{i+1}\rangle$ is a linear interpolation between the corresponding drawings and is also called a \emph{morphing step}, or simply a \emph{step}, for brevity.  In each step of a linear morph, each vertex of $G$ moves along a straight-line segment at a constant speed. 
Results~\cite{SODA-morph,angelini2013morphing, Barrera-unidirectional, Angelini-optimal-14} led up to a seminal paper by Alamdari et al.~\cite{alamdari2017morph} 
showing that for any two topologically equivalent planar drawings of a graph $G$, there is a linear $2D$ morph that transforms one drawing to the other in $\Theta(n)$ steps.
 This bound is asymptotically optimal in the worst case even when the graph $G$ is a path.  
A natural further question is how  the situation changes when we involve the third dimension. For general 3D graph drawings the problem seems challenging: 
it  is tightly connected to \emph{unknot recognition} problem, that is in $\NP \cap \coNP$~\cite{hlp-ccklp-99,l-ecktn-16}, and its containment in $\textbf{P}$ is wide open. 
If the given graph is a tree, the worst-case tight bound of $\Theta(n)$ steps holds for $3D$ crossing-free linear morph~\cite{PoleDance} (and the lower-bound example is again a path). 
If both the initial and the final drawing  are planar, then  $\mathcal{O}(\log n)$ steps suffice~\cite{PoleDance}. 

Both algorithmic results~\cite{alamdari2017morph,PoleDance}
have a drawback crucial  from the practical point of view. Their intermediate steps use infinitesimal or very small distances, as compared to distances in the input drawings. 
This may blow up the space requirements and affect the aesthetical aspect. 
This raises a demand for morphing algorithms that operate on a small grid, i.e., of size polynomial in the size of the graph and parameters of the input drawings. 
All the intermediate  drawings are then restricted to be \emph{grid drawings}, where vertices map to nodes of the grid.
Two crucial parameters of a straight-line grid drawing are: the area (or volume for the 3D case) of the required grid, and the \emph{resolution}, 
that is the ratio between the maximum edge length and the minimum vertex-edge distance. 
 If the grid area (or volume) is polynomially bounded, then so is the resolution~\cite{bblbfpr19}. 

Barrera-Cruz et al.~\cite{B-CHL18} gave an algorithm to morph between two straight-line planar Schnyder drawings of a triangulated graph in $\mathcal{O}(n^2)$ steps; all the intermediate drawings of this morph lie in a grid of size $\mathcal{O}(n) \times \mathcal{O}(n)$. Very recently  Barrera-Cruz et al.~\cite{bblbfpr19} gave an algorithm that linearly morphs between two planar straight-line grid drawings $\Gamma$ and $\Gamma'$ of an $n$-vertex rooted tree in $\mathcal{O}(n)$ steps while each intermediate drawing is also a planar straight-line drawing in a bounded grid. In particular, the maximum grid length and width is respectively $\mathcal{O}(D^3n \cdot L)$ and $\mathcal{O}(D^3n\cdot W)$, where $L=max\{l(\Gamma), l(\Gamma')\}$, $W=max\{w(\Gamma), w(\Gamma')\}$ and $D=max\{H, W\}$, $l(\Gamma)$ and $w(\Gamma)$ are the \emph{length} and the \emph{height} of the drawing $\Gamma$ respectively. Note that $D$ is $\Omega(\sqrt{n})$.

Here we study morphing one straight-line grid drawing $\Gamma$ of a tree
 to another such drawing $\Gamma'$ in \emph{sublinear} number of steps using three dimensions for intermediate steps. Effectively we ask the same question as in~\cite{PoleDance}, but now with the additional restriction that all drawings throughout the algorithm should lie in a small grid. 
We give two algorithms that require $\mathcal{O}(n)$ steps and  $\mathcal{O}(\sqrt{n} \log n)$ steps, respectively. 
All the intermediate drawings require 
 a $3D$ grid of length $\mathcal{O}(d^3(\Gamma) \cdot \log n)$,  width $\mathcal{O}(d^3(\Gamma) \cdot \log n)$ and height $\mathcal{O}(n)$, where $d(\Gamma)$ is the \emph{diameter} of the drawing $\Gamma$.
At the expense of using an extra dimension, we significantly decrease the number of morphing steps, and also the required (horizontal) area, as compared to the planar result~\cite{bblbfpr19}. 
Our algorithm is the first one to morph between planar drawings of a tree in $o(n)$ steps on a grid of polynomial volume.\par
In this work, we morph a $2D$ crossing free drawing of an $n$-vertex tree to the $3D$ canonical drawing of it and then analogously morph the  $3D$ canonical drawing to the final drawing. During the procedure, we use some standard techniques, e.g., canonical drawing~\cite{PoleDance}, rotation~\cite{B-CHL18}. Since a simple rotation on a fixed angle can not guarantee that the integral points map to the integral points, we generalize the concept of rotation to mapping. 

In Section~\ref{Sec:4}, we introduce a technique of lifting paths such that the vertices on the path along with their subtrees go the respective canonical positions and the drawing remains crossing free. The algorithm in Sec.~\ref{Sec:4} splits the given tree into disjoint paths that are lifted one by one in specific order by the described technique. 

In Section~\ref{sec:second_alg}, we introduce a technique of lifting a set of edges of the given tree. It is used for the algorithm that lifts the tree by dividing its edges into disjoint sets and lifting them one after another. We then merge the described algorithms together to produce an algorithm that uses $o(n)$ morphing steps.


\section{Preliminaries and Definitions} \label{sec:defs}


\paragraph{\textbf{Tree drawings}.}
For a tree $T$, let $r(T)$ be its root, and $T(v)$ be the subtree of $T$ rooted at a vertex $v$.  Let $V(T)$ and $E(T)$ be respectively the set of vertices and edges
 of $T$, and let $|T|$ denote the number of vertices in $T$. 

In a \textit{straight-line drawing}  of  $T$, each vertex $u$ is mapped to a point in $\mathbb{R}^2$ 
and each edge is mapped  to a straight-line segment connecting its end-points. 
A \emph{$3D$-} (respectively, a \emph{$2D$-}) \emph{grid drawing} of $T$ is a straight-line drawing where each vertex is mapped to a point with integer coordinates in $\mathbb{R}^3$ (respectively, $\mathbb{R}^2$). A  drawing  of $T$ is said to be \emph{crossing-free} if images of no two edges intersect except, possibly, at common end-points. 
A crossing-free $2D$-grid drawing is called a \emph{planar grid drawing}. For a crossing-free drawing $\Gamma$, let $\ball{\Gamma}{v}{r}$ denote the open disc
 of radius $r$ in the $XY_0$ plane centered at the image $\Gamma(v)$ of $v$. Unless stated otherwise, by the \textit{projection}, denoted by $pr()$, we mean the vertical projection to the $XY_0$ plane. 
Let $l(\Gamma)$, $w(\Gamma)$ and $h(\Gamma)$ respectively denote the \textit{length}, \textit{width} and \textit{height} of $\Gamma$, i.e., the maximum  absolute difference between the $x$-,$y$- and $z$-coordinates of vertices in $\Gamma$. Let  $d(\Gamma)$ denote the \textit{diameter} of $\Gamma$, defined as the ceiling of the maximum pairwise (euclidean) distance between its vertices. Given a vertex $v$ and an edge $e$ of the tree $T$, $\dist{\Gamma}{v}{e}$ denotes the distance between $\Gamma(v)$ and $\Gamma(e)$. Similarly, $\dist{\Gamma}{v}{u}$ is the distance between $\Gamma(v)$ and $\Gamma(u)$. For a  $3D$- (Also, a $2D$) grid drawing $\Gamma$, we define the resolution of $\Gamma$ as the ratio of the distances between the farthest and closest pairs of geometric objects where points
represent  images of nodes and line segments represent edges in $\Gamma$. 
We  define the vertical plane $Y=0$ (respectively, $X=0$)  by $XZ_0$ ($YZ_0$). We denote by $XZ^+_{0}$ ($XZ^-_{0}$) the vertical half-plane of $XZ_0$ with horizontal border-line going through the origin in $X$-positive (negative) direction. Similarly,  $XZ^+_{v}$ ($XZ^-_{v}$) is the shifted half-plane  $XZ^+_{0} (XZ^-_{0})$ where the origin is translated to the vertex $v$ and the horizontal border-line going through the vertex $v$.   $YZ^+_{v} (YZ^-_{v})$ denote the half-planes parallel to $YZ_0$ passing through $v$ in  $Y$-positive (negative) direction. We also define the horizontal plane passing through the origin by $XY_0$. We similarly define $XY^+_{0}$ and $XY^+_{v}$.

\begin{observation}
\label{obs:1}
 For $\Gamma$, $ M \leq d(\Gamma) \leq \sqrt{3}M$, where $M = \max(l(\Gamma), w(\Gamma), h(\Gamma))$. 
Thus 
to estimate the space 
required by $\Gamma$, it is enough to estimate  $d(\Gamma)$. 
\end{observation} 

\begin{lemma}
\label{lemma:pick}
\begin{enumerate}
	\item~\cite{bblbfpr19} For any vertex $v$ and edge $e$ not incident to $v$ in a planar grid drawing of $T$, $dist(v,e) \geq \frac{1}{d(\Gamma)}$.
	\item For each $v \in V(T)$, let $\Gamma_1(v)$ be equal to $(c \cdot \Gamma(v)_x, c \cdot \Gamma(v)_y)$ where $c$ is an integral constant. The  distance between any two distinct vertices in $\Gamma_1$ is at least $c$. 
\end{enumerate}
 
\end{lemma}

\begin{proof}
\begin{enumerate}
\item Let us fix a vertex $v \in V$ and an edge $e = (u, w) \in E$ such that $v \neq u, w$. Let $h$ be the distance between $v$ and the line spanned by $u,w$. height in the triangle $(v, u, w)$ from the vertex $v$,  Let $S$ be the area of triangle spanned by $u,v,w$.

Then, $S = \frac{1}{2} \cdot h \cdot dist_{\Gamma}(u, w) \leq \frac{d(\Gamma)}{2} \cdot h$.

On the other hand, from the Pick's theorem, we know that the area of the triangle with vertices lying in the lattice points of the grid is at least $\frac{1}{2}$.

Then, $\frac{1}{2} \leq S \leq \frac{d(\Gamma)}{2} \cdot h \Rightarrow h \geq \frac{1}{d(\Gamma)}$.\par

\item For any pair of distinct vertices $v_i, v_j$, either  $|\Gamma(v_i)_x - \Gamma(v_j)_x|$ or $|\Gamma(v_i)_y - \Gamma(v_j)_y|$ is at least $1$ since all vertices have integer coordinates in $\Gamma$.
 This implies that, for each pair of distinct vertices $v_i, v_j$, either  $|\Gamma_1(v_i)_x - \Gamma_1(v_j)_x|$ or $|\Gamma_1(v_i)_y - \Gamma_1(v_j)_y|$ is at least $c$ since each coordinate of each vertex in $\Gamma$ is multiplied by $c$ to obtain $\Gamma_1$. \qed
\end{enumerate}
\end{proof}

\begin{lemma}
\label{lemma:pick3d}
\begin{enumerate}
	\item For any vertex $v$ and edge $e$ not incident to $v$ in a $3D$-crossing free grid drawing of $\Gamma$ of $T$, $dist(v,e) \geq \frac{1}{d(\Gamma)}$.
	\item For a pair of non-adjacent edges $e_1, e_2$, the distance $dist(e_1,e_2) \geq \dfrac{1}{2\sqrt{3}\left(d(\Gamma)\right)^2}$ in a $3D$ crossing free grid drawing $\Gamma$ of $T$. 
\end{enumerate}
 
\end{lemma}

\begin{proof}
\begin{enumerate}
\item Let $e$ is spanned by the vertices $v_1,v_2$. If $v, v_1,v_2$ are co-linear, then $dist(v,e) \geq 1$ since each vertex is mapped to a point with an integral co-ordinate in $\mathbb{R}^3$. Let us assume that the points are not co-linear. Let the co-ordinates of points be $v=(x_0,y_0,z_0), v_1= (x_1,y_1,z_1)$ and $v_2=(x_2,y_2,z_2)$. The area of the triangles is $\dfrac{|\overrightarrow{v_0v_1}\times \overrightarrow{v_0v_2}|}{2}$, where $$\overrightarrow{v_0v_1}\times \overrightarrow{v_0v_2}=\begin{vmatrix}
\uvec{i} & \uvec{j} & \uvec{k} \\ 
x_1-x_0 & y_1-y_0 & z_1-z_0 \\ 
x_2-x_0 & y_2-y_0 & z_2-z_0 
\end{vmatrix}$$. Since all the points have integral co-ordinates and they a re not co-linear, $|\overrightarrow{v_0v_1}\times \overrightarrow{v_0v_2}| \geq 1$. This implies that the area  of the triangle spanned by $v,v_1,v_2$ is at least $\dfrac{1}{2}$. On the other hand, the area of the same triangle is  $\dfrac{1}{2}\cdot dist(v,e) \cdot dist(v_1,v_2) \leq \dfrac{1}{2}\cdot dist(v,e) \cdot d(\Gamma)$. This implies that $dist(v,e) \geq \dfrac{1}{ d(\Gamma)}$.
\item Let us assume that  $e_1$ is spanned by the vertices $v_1,v_2$ and $e_2$ is spanned by the vertices $v_3,v_4$. If these four points lie on the same plane, the problem maps to the previous one. We assume that these four points do not lie on a plane. Let the co-ordinates of points be $v_1= (x_1,y_1,z_1), v_2=(x_2,y_2,z_2), v_3=(x_3,y_3,z_3)$ and $v_4=(x_4,y_4,z_4)$. Then the volume of the tetrahedron spanned by these four points  is given by $\dfrac{dist(e_1,e_2)\cdot|\overrightarrow{v_1v_2}\times \overrightarrow{v_3v_4}|}{6}$. On the other hand the volume of the same tetrahedron is given by the following determinant.
$$\dfrac{1}{6}.
\left|~
\begin{vmatrix}
x_1-x_2 & y_1-y_2 & z_1-z_2 \\ 
x_2-x_3 & y_2-y_3 & z_2-z_3\\
x_3-x_4 & y_3-y_4 & z_3-z_4
\end{vmatrix}~\right|
$$

This implies that the volume of the determinant spanned by these four vertices is at least $\dfrac{1}{6}$, since the points have integral co-ordinates. This also implies that $dist(e_1,e_2) \geq \dfrac{1}{|\overrightarrow{v_1v_2}\times \overrightarrow{v_3v_4}|} \geq \dfrac{1}{2\sqrt{3}\left(d(\Gamma)\right)^2}$.

\end{enumerate}
\end{proof}


\paragraph{\textbf{Path decomposition}}

$\mathcal{P}$ of a tree $T$ is a decomposition of its edges into a set of disjoint paths as follows. Choose a path in $T$ and store it in $\mathcal{P}$. Removal of this path may disconnect the tree; recurse on the remaining connected components. Note  that during the removal of a path, we only delete its edges, not the vertices. In the end, $\mathcal{P}$ contains disjoint paths whose union is $T$. The way of choosing paths  may differ. The depth of a vertex $v$ in $T$, denoted by $dpt(v)$, is defined as the length of the path from $r(T)$ to $v$. \textit{Head} of the path $P$, denotes as $head(P)$, is the vertex $x \in P$ with the minimum depth in tree $T$. Let the \textit{internal vertices} of the path be all vertices except $head(P)$. 
Any path decomposition $\mathcal{P}$ of $T$ induces a linear order of the paths: path $P'$ succeeds $P$, i.e., ${P' \succ P}$, if and only if $P'$ is deleted before $P$
 during the construction of $\mathcal{P}$. Note that the subtree of each internal vertex of a path $P$ is a subset of the union  of the  paths that precede $P$. 

In the \textit{long-path decomposition}~\cite{bender_laps}, the path chosen in every iteration is the longest root-to-leaf path (ties are broken arbitrarily). Let $\mathcal{L} = \{L_1, \ldots, L_m\}$ be the ordered set 
of paths of a long-path decomposition of $T$. For $i <j$,  $|L_i| \leq |L_j|$.


\textbf{Heavy-rooted-pathwidth decomposition~\cite{PoleDance}.} The \textit{Strahler number} or \textit{Horton-Strahler number} of a tree is a parameter which was introduced by Horton and Strahler~\cite{h-ed-45, s-ha-52, s-ha-57}. The same parameter was recently rediscovered by Biedl~\cite{therese} with the name of \textit{rooted pathwidth} when addressing the problem of computing upward tree drawings with optimal width.




\begin{figure*}[!ht]
  \centering
  \subfloat[]{\centering \includegraphics[scale=0.7]{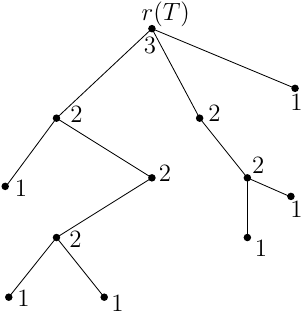}\label{fig:rpw}}
  \centering
  \subfloat[]{\centering \includegraphics[scale=0.7]{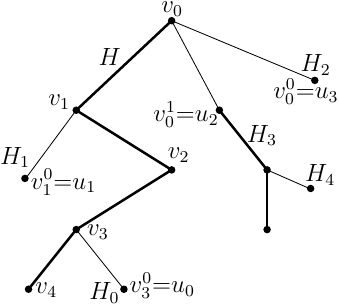}\label{fig:heavy-path-2}}
  \centering
  \subfloat[]{\centering \includegraphics[scale=1]{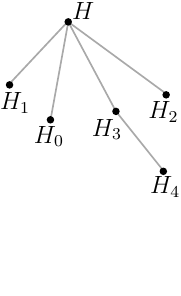}\label{fig:heavy-path-tree}}
  \caption{~\cite{PoleDance} The illustration (a) shows, for each vertex $v$ of a tree $T$, the number $rpw(T(v))$. In particular, $rpw(T) = 3$. The illustration (b) shows with bold lines the heavy edges of $T$ forming the heavy paths $H, H_0, \ldots, H_4$. The illustration (c) shows the path tree of $T$}
\end{figure*}

The \textit{rooted pathwidth} of a tree $T$, which we denote by $rpw(T)$, is defined as follows. If $|T| = 1$, then $rpw(T) = 1$. Otherwise, let $k$ be the maximum rooted pathwidth of any subtree rooted at a child of $r(T)$. Then $rpw(T) = k$ if exactly one subtree rooted at a child of $r(T)$ has rooted pathwidth equal to $k$, and $rpw(T) = k + 1$ if more than one subtree rooted at a child of $r(T)$ has rooted pathwidth equal to $k$; see Figure~\ref{fig:rpw}. Clearly $rpw(T)$ is an integer number.

The \textit{heavy-rooted-pathwidth decomposition} of a tree $T$ is defined as follows refer to Figure~\ref{fig:heavy-path-2}. For each non-leaf vertex $v$ of $T$, let $c^*$ be the child of $v$ in $T$ such that $rpw(T(c^*))$ is maximum (ties are broken arbitrarily). Then $(v, c^*)$ is a heavy edge; further, each child $c \neq c^*$ of $v$ is a \textit{light child} of $v$, and the edge $(v, c)$ is a \textit{light edge}. Connected components of heavy edges form set of paths $\mathcal{H}(T) = \{H, H_0, \ldots, H_k\}$, called \textit{heavy paths}, which may have many incident light edges. The \emph{path tree} of $T$ is a tree whose vertices correspond to heavy paths in $T$; see Fig.~\ref{fig:heavy-path-tree}. The parent of a heavy path $P$ in the path tree is the heavy path that contains the parent of the vertex with minimum depth in $P$. The root of the path tree is the heavy path containing $r(T)$.

For convenience, we will also consider the light edge immediately above a heavy path to be a part of the path.  When it is clear from the context we will refer to $\mathcal H(T)$ simply as $\mathcal{H}$.  

We denote by $H$ the root of the path tree of $T$; let $v_0, \ldots, v_{m - 1}$ be the ordered sequence of the vertices of $H$, where $v_0 = r(T)$. For $i = 0, \ldots, m - 1$, we let $v^0_i, \ldots, v^{t_i}_i$ be the light children of $v_i$ in any order. Let $L = u_0, u_1, \ldots, u_{l - 1}$ be the sequence of the light children of $H$ ordered so that: (i) any light child of a vertex $v_j$ precedes any light child of a vertex $v_i$, if $i < j$; and (ii) the light child $v^{j + 1}_i$ of a vertex $v_i$ precedes the light child $v^j_i$ of $v_i$. For a vertex $u_i \in L$, we denote by $p(u_i)$ its parent; note that $p(u_i) \in H$. It is known~\cite{therese}  that the height of the path tree of an $n$-vertex tree $T$ is at most $rpw(T) \in \mathcal{O}(\log n)$. \\

Fig.~\ref{fig:canonical_drawing} shows the heavy-rooted-pathwidth decomposition of our running example where heavy paths are shown in different colors. 

\emph{The canonical 3D drawing} of a tree $T$, introduced in~\cite{PoleDance}, is the crossing-free straight-line 3D drawing of $T$ that maps each vertex $v$ of $T$ to its \textit{canonical position} $\mathcal{C}(v)$ defined by the heavy-rooted pathwidth decomposition.

The canonical drawing of a tree $T$, denoted by $\mathcal{C}(T)$, is defined as follows:

\begin{itemize}
	\item First, we set $\mathcal{C}(v_0) = (0, 0, 0)$ for the root $v_0$ of $T$.
	\item Second, for each $i = 1, \ldots, k-1$, we set $\mathcal{C}(v_i) = (0, 0, z_{i-1} + |T(v_{i-1})| - |T(v_i)|)$, where $z_{i-1}$ is the $z$-coordinate of $\mathcal{C}(v_{i-1})$.
	\item Third, for each $i = 1, \ldots, k - 1$ and for each $j = 0, \ldots, t_i$, we determine $\mathcal{C}(v^{j}_i)$ as follows. If $j = 0$, then we set $\mathcal{C}(v^j_i) = (1, 0, 1 + z_i)$, where $z_i$ is the $z$-coordinate of $\mathcal{C}(v_i)$; otherwise, we set $\mathcal{C}(v^j_i) = (1, 0, z^{j-1}_i + |T(v^{j-1}_i)|)$, where $z^{j-1}_i$ is the $z$-coordinate of $\mathcal{C}(v^{j-1}_i)$.
	\item Finally, in order to determine the canonical positions of the vertices in $T(v^j_i) \setminus \{v^j_i\}$, for each $i = 0, \ldots, k - 2$ and each $j = 0, \ldots, t_i$, we recursively construct the canonical 3D drawing  of $T(v^j_i)$, and translate all the vertices by the same vector so that $v^j_i$ is sent to $\mathcal{C}(v^j_i)$.
\end{itemize}

We  use the fact that $\mathcal{C}(T)$ lies in the positive the $XZ_0$ quarter-plane, inside a bounding box of height $|T|$ and width $rpw(T)$ such that the left-bottom point of the box represents the image of the root $v$ and is mapped to the origin. Since tree $T$ never changes throughout our algorithm, we refer $rpw(T)$ as  $rpw$.

\begin{figure*}[!ht]
  \centering
  \subfloat[]{\centering \includegraphics[scale=0.5]{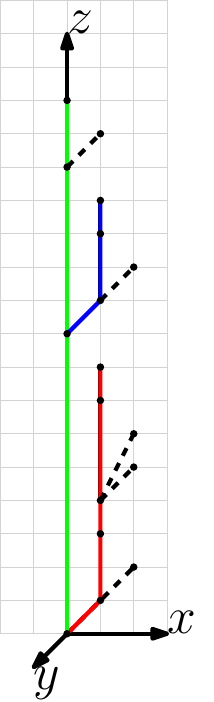}\label{fig:canonical_drawing}}
  \centering
  \subfloat[]{\centering \includegraphics[scale=0.5]{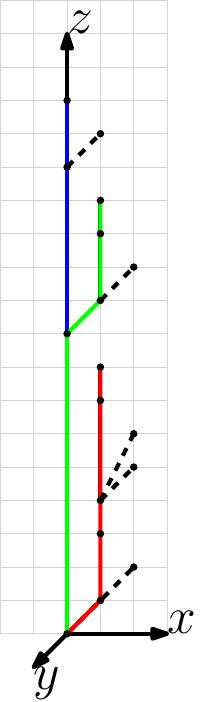}\label{fig:canon2}}

  \caption {An example of the canonical drawing of a tree. Illustrations (a) and (b) show heavy and long paths respectively with different colors.}
\end{figure*}

For any vertex $v$ of $T$, the \emph{relative canonical drawing} $\mathcal{C}_{T_v}$ of $T(v)$ is the drawing of $T(v)$ obtained by cropping it from $\mathcal{C}(T)$ and translating the obtained drawing of $T(v)$ so that $v$ is mapped to the origin. 
 
\paragraph{\textbf{Morph:}}

Let $\Gamma$ and $\Gamma'$ be two planar straight-line drawings of $T$. Then a morph $\mathcal{M}$ is a sequence $\langle \Gamma_1, \Gamma_2,\ldots, \Gamma_k \rangle$ of 3D straight-line drawings of $T$ such that $\Gamma_1 = \Gamma, \Gamma_k = \Gamma'$, and $\langle \Gamma_i, \Gamma_{i+1} \rangle$ is a linear morph, for each $i = 1, \ldots, k - 1$. 

\section{Tools for morphing algorithms}
\label{sec:Tools}
\paragraph{\textbf{Stretching  with a constant $\mathcal{S}_1$.}} 
\label{sec:stretching}

Let the drawing $\Gamma$ be lying in the the $XY_0$ plane. During the \emph{stretching morph} each coordinate of each vertex in $\Gamma$ is multiplied by a common positive integer constant $\mathcal{S}_1$. Thereby, it is a linear morph that "stretches" the vertices apart.

\begin{lemma} 
\label{lemma:stretching_morph_is_crossing_free}
The stretching morph is a crossing-free morph.
\end{lemma}

\begin{proof}
 In the drawing  $\Gamma^t$, for  $t \in [0, 1]$,  the image of a vertex $v$ of $T$ is  $\{t\Gamma_1(v) + (1 - t)\Gamma(v)\}$. This implies that it is the same drawing $\Gamma$ scaled by $t\mathcal{S}_1 + (1 - t)$. Since the original drawing $\Gamma$ was crossing-free, so is the drawing $\Gamma^t$. \qed
\end{proof}



\begin{lemma} 
 \label{lemma:circle_around}
\begin{enumerate}
	\item For any pair  $v_i, v_j$ of vertices,  disks \ball{\Gamma_1}{v_i}{\frac{\mathcal{S}_1}{2}} and \ball{\Gamma_1}{v_j}{\frac{\mathcal{S}_1}{2}} do not cross in the $XY_0$ plane.
	\item For a vertex $v_i$ in $\Gamma_1$, disk \ball{\Gamma_1}{v_i}{\frac{\mathcal{S}_1}{2 \cdot d(\Gamma)}} does not enclose any other vertices or any part of edges non-incident to $v_i$.
	\item For every vertex $v$ and every edge $e = (v, u)$ in $\Gamma_1$, there is lattice point $z$ such that $z \in e$ and $z \in \ball{\Gamma_1}{v_i}{d(\Gamma)}$.
\end{enumerate}
\end{lemma}

\begin{proof}
\begin{enumerate}
	\item Since the disk \ball{\Gamma}{v_i}{1} contains no other vertices $v_j, j \neq i$, then \ball{\Gamma_1}{v_i}{\mathcal{S}_1} does not contain other vertices too. The \ball{\Gamma_1}{v_i}{\frac{\mathcal{S}_1}{2}}, i.e., the disk with the radius $\frac{\mathcal{S}_1}{2}$, does not intersect with any \ball{\Gamma_1}{v_j}{\frac{\mathcal{S}_1}{2}} (for $j \neq i$) as $\mathcal{S}_1$ is the stretching factor.

	\item Due to Lemma~\ref{lemma:pick},  no non-incident edges intersects the disk of radius $\frac{1}{d(\Gamma)}$ around $v_i$  in $\Gamma$. That means that in $\Gamma_1$, where all distances are multiplied by $\mathcal{S}_1$, no non-incident edges intersect the disk of radius $\frac{\mathcal{S}_1}{d(\Gamma)}$ around $v_i$. Other vertices do not lie in \ball{\Gamma_1}{v_i}{\frac{\mathcal{S}_1}{2 \cdot d(\Gamma)}} as $\ball{\Gamma_1}{v_i}{\frac{\mathcal{S}_1}{2 \cdot d(\Gamma)}} \subset \ball{\Gamma_1}{v_i}{\frac{\mathcal{S}_1}{2}}$.

	\item Every edge $e = (u, v)$ from $\Gamma_0$ has been stretched by $\mathcal{S}_1$.  If $\Gamma_0(u) = (u_x, u_y)$ and $\Gamma_0(v) = (v_x,v_y)$, then $\Gamma_1(u) = (u_x \cdot \mathcal{S}_1, u_y \cdot \mathcal{S}_1)$ and $\Gamma_1(v) = (v_x \cdot \mathcal{S}_1, v_y \cdot \mathcal{S}_1)$. The integral point $z \in e$ and $z \in \ball{\Gamma_1}{v}{d(\Gamma)}$ is $(u_x \cdot \mathcal{S}_1 + v_x - u_x, u_y \cdot \mathcal{S}_1 + v_y - u_y)$.\qed

\end{enumerate}
\end{proof}

\paragraph{\textbf{Mapping around a pole}}
\label{sec:mapping_pole}
 Let the \textit{pole through $(x', y')$} be  vertical line in $3D$  through point $(x', y', 0)$. Let $\alpha, \beta$ be vertical half-planes containing the pole $l$ through a point  with integer coordinates. Suppose $\angle(\alpha, \beta) \notin \{ 0, \pi\}$ and $\alpha, \beta$ contain infinitely many points with integer coordinates. Mapping \emph{around the pole $l$} is a morphing step to obtain a drawing $\Gamma'$ which lies in  $\beta$ from  $\Gamma$ which lies in $\alpha$. Each vertex moves along a horizontal vector between $\alpha$ and $\beta$. The direction of this vector is common for all vertices of $\Gamma$ and is defined by  $\alpha$ and $\beta$. Let us fix a horizontal plane $h$ which passing through the point $(0,0,b)$ where $b$ is an integer. Let   $p_{\alpha}, p_{\beta}$ be points that lie on $h\cap \alpha$ and  $h\cap \beta$, respectively; such that $dist(l, p_{\alpha}) = d_{\alpha}$ and $dist(l, p_{\beta}) = d_{\beta}$ be the minimum non-zero distances from the $l$ to the integer points lying in $h\cap \alpha$ and  $h\cap \beta$.  The  \textit{vector of mapping} $u$ is defined as $\frac{p_{\beta} - p_{\alpha}}{|p_{\beta} - p_{\alpha}|}$. Mapping is a unidirectional morph since all vertices of $\Gamma$ move along the vectors parallel to the vector of mapping till they reach the half-plane $\beta$. See Fig.~\ref{fig:mapping}.
 
\begin{figure*}[!ht]
  \centering 

  \includegraphics[scale=0.65]{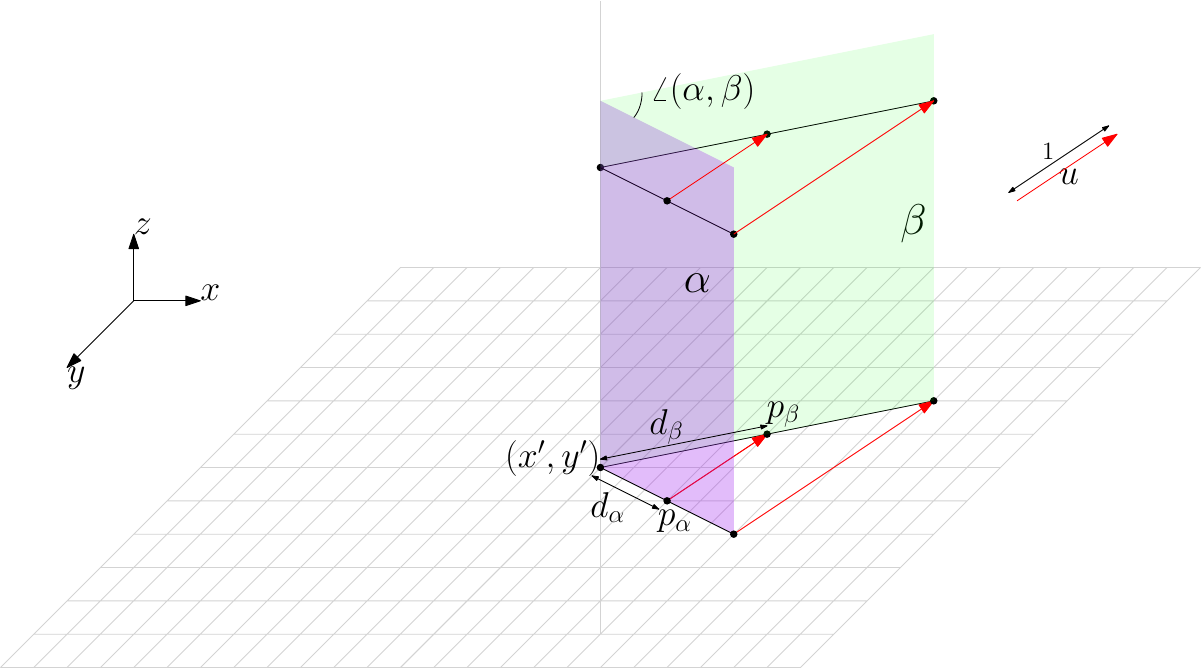}
  
  \caption {The illustration shows the mapping morph, half-planes $\alpha, \beta$ sharing a common pole through point $(x', y')$ and their vector of mapping.}
  
  \label{fig:mapping}
\end{figure*}

\begin{lemma} \label{lemma:mapping_crossing-free}
Mapping is a crossing-free morph.
\end{lemma}

\begin{proof}
Let $\alpha, \beta$ be vertical half-planes containing the pole $l$ through the origin. Let us assume that $\angle(\alpha, \beta)=\theta$ and  let us define the stretching factor $S_m = \dfrac{d_{\beta}}{d_{\alpha}}$. Then, we can define mapping by the following matrix. 
$$
\begin{pmatrix}
S_m \cos \theta & -S_m \sin \theta &  0 \\
S_m \sin \theta & S_m \cos \theta  &  0 \\
0 & 0  & 1 
\end{pmatrix}
$$
Note that the determinant of the matrix is positive. This implies that the order type of the points does not change during mapping. Since the drawing in $\alpha$ is crossing free, the drawing obtained in $\beta$ is also crossing free. Note that if the pole passes through the point $(a,b)$ instead of the origin, we can first translate the point $(a,b)$ to origin. We again translate back after mapping. 

\end{proof}

 The \textit{rotation} is a mapping  when $\alpha, \beta$ are half-planes of planes parallel to the $XZ_0, YZ_0$ respectively and thus $\angle(\alpha, \beta) = \frac{\pi}{2}$. A \textit{horizontal pole} is a line parallel to $0X$ axis. We similarly define a horizontal rotation which maps $XZ_0$ to $XY_0$ plane.  

\begin{figure}[!htp]

  \begin{center}  
  \includegraphics[scale = 0.6]{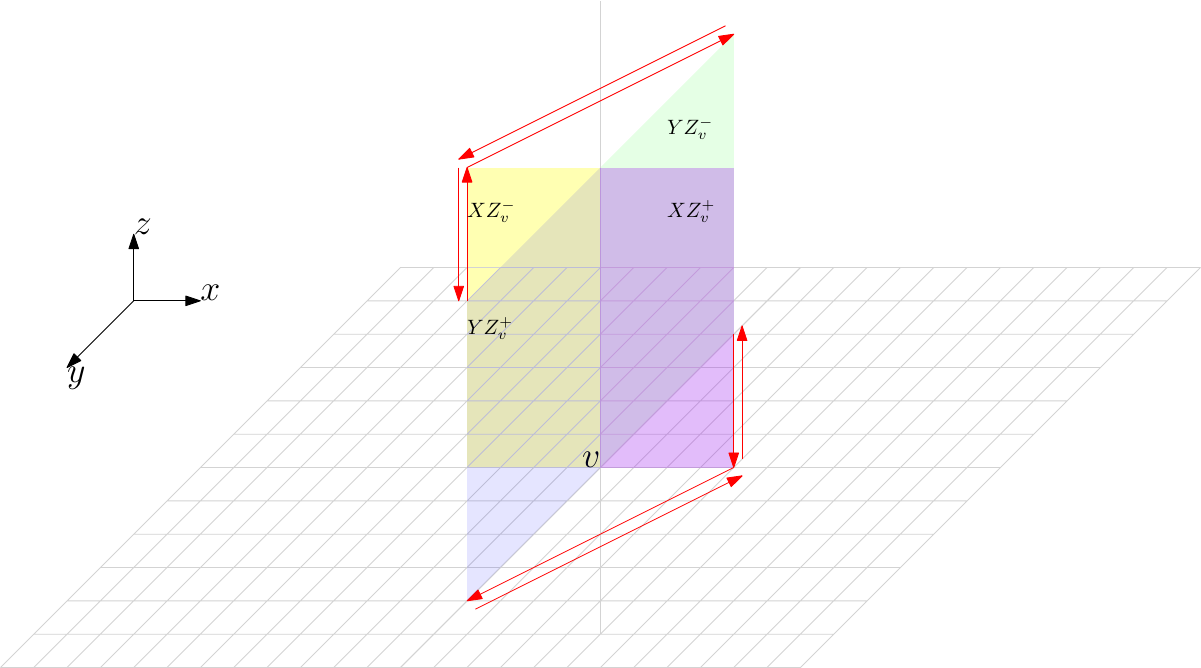}
  \caption{Possible vectors of rotation between 4 half-planes sharing a common pole through vertex $v$}
  \label{fig:rotation}
  \end{center}

\end{figure}

\paragraph{\textbf{Shrinking lifted subtrees.}}
\label{sec:shrink_subtrees}
 Let $v$ be a  vertex of $T$.  Suppose that the image of subtree $T(v)$ in $\Gamma(T)$ is $\mathcal{C}_{T_v}$, i.e., it lies in $h = XZ^+_v$.  Let  $\mathscr{C}=\{v_1, \ldots, v_{l}\}$ be the  set of the children of $v$, ordered according to their $z$-coordinates in $\mathcal{C}_{T_v}$. Let $\mathscr{C'}= \{v_{i_1}, \ldots, v_{i_k}\}$ be an ordered subset of $\mathscr{C}$ such that  the members of $\mathscr{C'}$ having same order among them as in $\mathscr{C}$.  Let us consider the new subtree $T'(v)$ which is obtained by deleting the vertices in $\mathscr{C}\setminus \mathscr{C'}$ and their subtrees. The  drawing of $T'(v)$  obtained by deleting the vertices in $\mathscr{C}\setminus \mathscr{C'}$ and their subtrees from $\mathcal{C}_{T_v}$ is denoted by $\mathcal{C}_{T_v}(T'_v)$, i.e., the canonical drawing of $T(v)$ restricted to $T'(v)$. Note that any vertex of $T'(v)$ has same coordinate  in $\mathcal{C}_{T_v}$ and $\mathcal{C}_{T_v}(T'_v)$.  Also, for each $j$ satisfying $1\leq j \leq k$, $T'(v_{i_j})$ lies inside a box of height $|T(v_{i_j})|$ and width $rpw(T(v_{i_j}))$ on  $h$. We define the \emph{shrink subtree} procedure on $\mathcal{C}_{T_v}(T'_v)$ as follows. We move each vertex $v_{i_j}$ along with their subtrees from $\mathcal{C}_{T_v}(v_{i_j})$ to $(\mathcal{C}_{T_v}(v_{i_j})_x, \mathcal{C}_{T_v}(v_{i_j})_y, \mathcal{C}_{T'_v}(v_j)_z)$.  Let us denote the shrunk subtree by $\mathcal{C'}_{T'_v}$. The  height of the shrunk subtree $\mathcal{C'}_{T'_v}$  is equal to the  number of  vertices in $T'(v)$. 
 
\begin{figure*}[!ht]
  \centering 
  
  \includegraphics[scale=0.6]{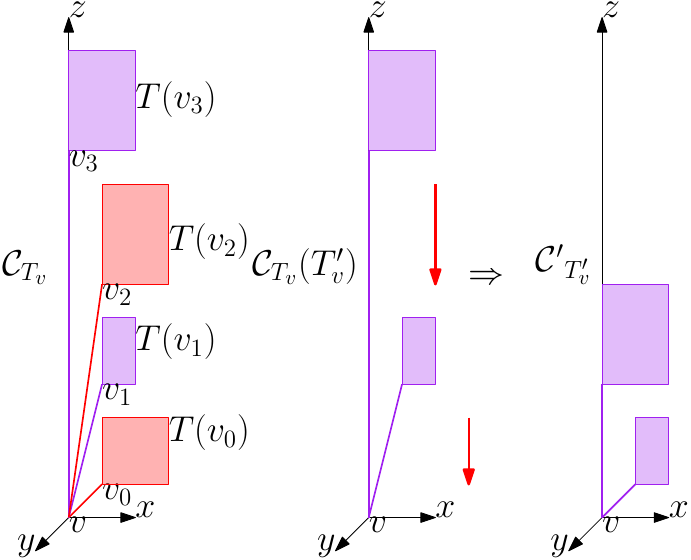}
  
  \label{fig:shrinking}

  \caption {In the drawing there is an example of shrinking morph when $l = 4$.}
\end{figure*}
 
 \begin{observation}\label{obs:shrink_height}
  The  height of the shrunk subtree $\mathcal{C'}_{T'_v}$  is equal to the  number of  vertices in $T'(v)$. 
   \end{observation}

 \begin{proof}
  Note that the subtree rooted at $v$ has a $3D$ canonical drawing in $\mathcal{C}_{T_v}(T'_v)$.
  This implies that for $1\leq j \leq k$ the subtree rooted at $v_{i_j}$ also has the canonical drawing with respect to $v_{i_j}$ in $\mathcal{C}_{T_v}(T'_v)$.
  Each  $T'(v_{i_j})$ lies inside a box of height $|T'(v^{i_j})|$. Also, note that the z-coordinate of $ \mathcal{C'}_{T'_v}(v_{i_j})$ is equal to $\mathcal{C}_{T_v}(v_j)_z$. This implies that in $\mathcal{C'}_{T'_v}$,  $T'(v)$ lies inside a box of height $\sum_{j=1}^{k} |T(v_{i_j})|+1=|T'(v_i)|$ and the down-left corner of the box coincide with $\mathcal{C}_{T_v}(v)=\mathcal{C'}_{T'_v}(v).$ Also note that $|T'(v)| \le |T(v)|$.
 \end{proof}
  \begin{proposition}~\cite{alamdari2017morph}
 \label{cor:unidirectional}
 Let $\langle \Gamma, \Gamma' \rangle$ be a unidirectional morph between two planar straight-line drawings $\Gamma$ and $\Gamma'$ of $T$. Let $u$ be a vertex of $T$, let $vw$ be an edge of $T$ and, for any drawing of $T$, let $l_{vw}$ be the line through the edge $vw$ oriented from $v$ to $w$. Suppose that $u$ is to the left of $l_{vw}$ both in $\Gamma$ and in
$\Gamma'$. Then $v$ is to the left of $l_{vw}$ throughout $\langle \Gamma, \Gamma' \rangle$.
 \end{proposition} 
\begin{lemma}
  \label{lemma:shrinking}
 Shrinking is a crossing-free morph.
 \end{lemma}
 
 \begin{proof}
 Note that the $x$ and $y$ coordinates of each vertex  in $\mathcal{C}_{T'_v}$ and $\mathcal{C'}_{T'_v}$ are the same and $z$-coordinate can only decrease, shrinking is obviously a unidirectional morph in vertical $XZ$ plane which passes through $v$. Also, note that through out the morphing process all vertices  in $T'_v$ maintain the relative orders among themselves since they move along parallel vectors. Also note that shrinking satisfies the conditions of Proposition~\ref{cor:unidirectional}. This implies that shrinking is a crossing free morph. \qed
 \end{proof}

\paragraph{\textbf{Turning in the horizontal plane.}}
\label{sec:turning_in_horiz}

Let $\Gamma_0(T(v))$ be the canonical drawing of a subtree $T(v)$ on the horizontal plane $\alpha$, i.e. relative canonical drawing $\mathcal{C}_{T_v}$ rotated around horizontal pole through $v$ in any direction to lie in $XY_{v}^{+}$ or $XY_{v}^{-}$ half-plane. We discuss the case when $\Gamma_0(T(v))$ lies in $XY_{v}^{+}$, the other case is similar. Let $\Gamma_1(T(v)), \Gamma_2(T(v)), \Gamma_3(T(v))$ be the drawings  obtained from  $\Gamma_0(T(v))$ by rotating the horizontal plane $\alpha$ around the point $\Gamma(v)$ by the angles $\frac{\pi}{2}, \pi, \frac{3\pi}{2}$ respectively. Using a lemma from~\cite[Lemma~pinwheel]{bblbfpr19}, we show the following. 

\begin{figure}[!htp]

  \begin{center}  
  \includegraphics[scale = 0.5]{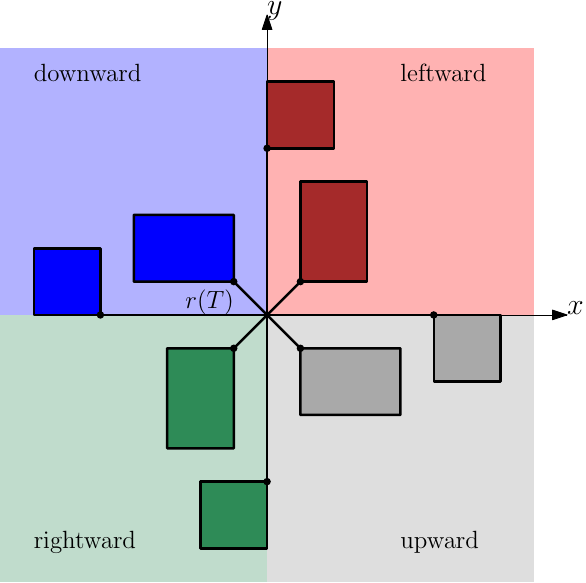}
  \caption{Possible orientations of the canonical drawing of a tree $T$.}
  \label{fig:pinwheel}
  \end{center}
\end{figure}

\begin{lemma} (Pinwheel)~\cite{bblbfpr19} 
\label{lemma:pinwheel}
Let $\Gamma$ and $\Gamma'$ be two canonical drawings of a rooted ordered tree $T$ , where $r(T)$ is at the same point in $\Gamma$ and $\Gamma'$. If $\Gamma$ and $\Gamma'$ are (i) upward and leftward, or (ii) leftward and downward, or (iii) downward and rightward, or (iv) rightward and upward, respectively, then the morph $\langle \Gamma, \Gamma' \rangle$ is planar and lies in the interior of the right, top, left, or bottom half of the $2n$-box centered at $r(T)$, respectively. See Fig.~\ref{fig:pinwheel}.
\end{lemma}

\begin{lemma}
\label{lemma:pinwheel_ref}
The drawing $\Gamma_i(T(v))$ can be obtained from the drawing $\Gamma_{i - 1}(T(v))$ in one morphing step, $i = 0, 1, 2, 3$, where by $\Gamma_{0 - 1}(T(v))$ we mean $\Gamma_3(T(v))$.
\end{lemma}

\begin{proof}
Proof of this fact is the same as the proof of the Lemma~\ref{lemma:pinwheel} of~\cite{bblbfpr19}.\qed 
\end{proof}

\section{Morphing through lifting paths}
\label{Sec:4}

Let $T$ be an $n$-vertex tree  and $\mathcal{P}$ be a  path decomposition of $T$ into $k$ paths. In this section, we describe an algorithm that morphs a plane drawing $\Gamma = \Gamma_0$ in the $XY_0$ plane of tree $T$ to the canonical 3D drawing $\Gamma' =\mathcal{C}(T)$ of $T$ in $\mathcal{O}(k)$ steps. It lifts the paths of $\mathcal{P}$ one by one applying the procedure \lift{}.   Note that the final positions for the vertices in $\mathcal{C}(T)$ are defined through heavy-rooted-pathwidth decomposition, and do not depend on $\mathcal{P}$. Also, a morph from $\mathcal{C}(T)$ to $\Gamma'$ can be obtained by playing the morph from $\Gamma'$ to $\mathcal{C}(T)$ backwards.  At all times during the algorithm, the following invariant holds: a path $P_i \in \mathcal{P}$ is lifted only after all the children of  the internal vertices of $P_i$ are  lifted. After the execution of \lift{P_i}, path $P_i$ moves to its canonical position with respect to $head(P_i)$. See Fig.~\ref{fig:step0}-\ref{fig:step13}. \\

\noindent\textit{\textbf{Step 0: Preprocessing.}} This step is a single stretching morph $\langle \Gamma, \Gamma_1 \rangle$ with $\mathcal{S}_1 = 2 \cdot (rpw + d(\Gamma))$, see Section~\ref{sec:stretching}, stretching is a crossing-free morph.

\subsection{Procedure \lift{P}}

\begin{figure}[!htp]

  \begin{center}  
  \includegraphics[scale = 0.6]{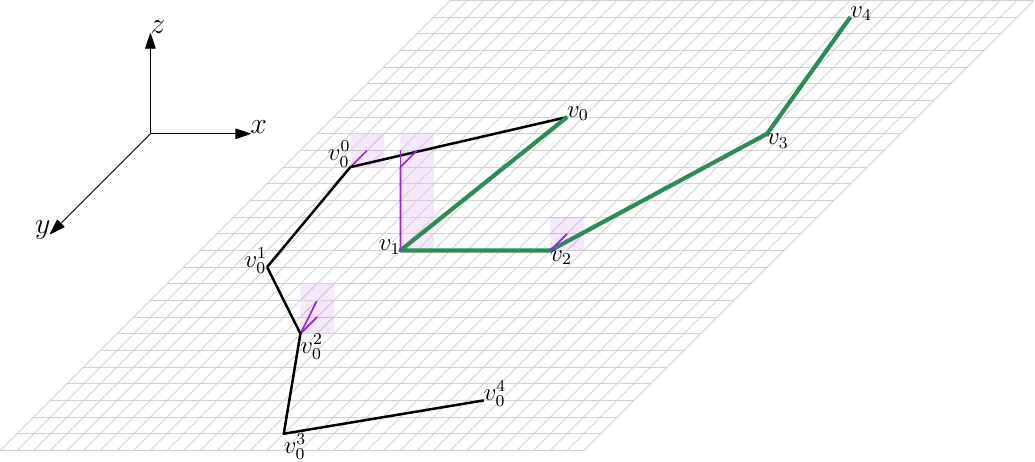}
  \end{center}
  
  \caption{\textbf{Drawing $\Gamma_t$}. In the picture there is the drawing $\Gamma_t$ after lifting all paths containing purple edges. All lifted subtrees are shown in light purple.}
  \label{fig:step0}

\end{figure}

Let $ P_i = (v_0, v_1, \ldots, v_m)$ be the first path in $\mathcal{P}$ that has not been processed yet and $\Gamma_t$ be the current drawing of $T$. We lift the path $P_i$. For any vertex $v$ let \textit{lifted subtree} $T'(v)$ be the portion of subtree $T(v)$ that has been lifted after execution of \lift{P_j} for some $j < i$. The following invariants hold after every iteration of \lift{} procedure. 

\begin{itemize}
	\item[(I)] the drawing of $T'(v)$ in $\Gamma_t$ is the canonical drawing of $T'(v_j)$ with respect to $v$ for any $v \in V(T)$,
	\item[(II)] vertices of paths $P_k, k > i$ lie within the $XY_0$ plane.
\end{itemize} 

We prove that these conditions hold after lifting $P_i$. Let the \textit{processing vertices} be the internal vertices of $P_i$ along with the vertices of their lifted subtrees.

\begin{lemma} The subtrees of all internal vertices $v_j$ in $P_i$ are already lifted. 
\end{lemma}

\begin{proof}
All the paths that precede $P_i$ in $\mathcal{P}$ are exactly the paths that are left after deleting $P_i$ from the forest of subtrees of $T$ (see paragraph \textit{Path decomposition} in Section~\ref{sec:defs}). From this fact it follows that all internal vertices of $P_i$ have become roots of some trees after deleting the edges of $P_i$ and edges of that trees were processed in \lift{P_j} procedures for some $P_j, j < i$. \qed
\end{proof}
\begin{lemma}
\label{lemma:n-height}
The maximum height of vertices in the drawing $\Gamma_t$ is strictly less than  $n$.
\end{lemma}

\begin{proof}
As all vertices of the processed paths are part of lifted subtrees of their head vertices, and by (I) have height as in the canonical drawing, their height is strictly less than the number of vertices in a corresponding subtree, which does not exceed number of vertices in $T$. \qed
\end{proof}
\begin{lemma}
\label{lemma:rpw_dist}
For any vertex $v$ and its lifted subtree $T'(v)$. The maximum horizontal distance in the lifted subtree $T'(v)$ of $v$, i.e. the difference between $x$-coordinates between any pair of vertices in $T'(v)$ in $\Gamma_t$, does not exceed $rpw$.
\end{lemma}

\begin{proof}
Note that the drawing of $T'(v)$ in $\Gamma_t$ is the canonical drawing of $T'(v)$ with respect to $v$. This implies that the maximum horizontal distance between any vertex in $T'(v)$ and $v$ in  $\Gamma_t$ is at most $rpw$. \qed
\end{proof}

\paragraph{\textbf{Overview}}

The procedure \lift{P_i} consists of 13 steps and results in moving vertices of the path $P_i$ along with their lifted subtrees to their canonical positions with respect to the head of $P_i$, i.e.,vertex $v_0$. Note that the preprocessing Step 0 and Lemmas.~\ref{lemma:circle_around},~\ref{lemma:n-height} and~\ref{lemma:rpw_dist} guarantee that the already lifted subtrees lie in the disjoint cylinders of radius $rpw$ and height $n$.

Step 1 is a preparatory step and is needed to ensure that the maximum height of a vertex does not exceed $2n$ during the \lift{.} procedure. Step 1 is a simultaneous shrinking morph for already lifted subtrees of the internal vertices, and it is a crossing free morph since it happens in the corresponding disjoint cylinders. 

Step 2 is also preparatory and is needed to make sure, that no intersections happen during Steps 3. In this step, lifted subtree of vertex $v_j (j < m)$ of $P_i$ is rotated to such a direction  that in the projection it does not overlap with the edge $(v_j, v_{j + 1})$.

In Step 3, all the internal vertices along with their lifted subtrees are moved vertically  such that height that all $T'(v_j)$(lifted subtrees of internal vertices) become horizontally separated from each other and from the non-processed part of the tree. Note that the vertices of $P_i$ are in the same vertical order as in the $\mathcal{C}(T)$. 

In Step 4, lifted subtree of each of the internal vertices is rotated to lie in a horizontal plane through the corresponding vertex. This step places all $T'(v_j)$ in disjoint horizontal planes. 

After that we want to get $\mathcal{C}(T(v_1))$, i.e., the canonical drawing of all the subtree of $v_1$, translated to the current position of $v_1$.  During Steps 5-9 we do not move vertex $v_1$ and other non-processing vertices.

Steps 5 and 6 move $v_2, \ldots, v_m$ to their canonical positions with respect to the vertex $v_1$ (Step 5 corrects $x, y$-coordinates and Step 6 corrects $z$-coordinate). During both of these steps every lifted subtree of $v_2, \ldots, v_m$ moves along the same vector as its root vertex.

Steps 7, 8 and 9 place lifted subtrees of the internal vertices to the canonical positions with respect to $v_1$. In the canonical drawing, every subtree lies in $x+$ direction from its root. Step 7 simultaneously turns (if needed) every subtree of the internal vertices in $y+$ direction. Then in Step 8, we need to make the inverse step to Step 1 and morph $T'(v_j)$ from shrunk drawing to the canonical one. In Steps 7 and 8, all vectors of movements are horizontal and no intersections can happen due to disjointness of the corresponding horizontal planes. After Step 8 every lifted subtree of internal vertices is already in canonical position with respect to its root, but lying in $XY^+$ instead of $XZ^+$. In Step 9, we rotate $T'(v_j)$ to the vertical plane simultaneously for every internal vertex $v_j$. 

Steps 10-13 move $v_1$ with its subtree $T(v_1)$ and (if needed) $T'(v_0)$ to its canonical position with respect to $v_0$. Step 10 horizontally moves $T(v_1)$ towards the vertical pole through $v_0$ until $T(v_1)$ lies within the cylinder of radius $rpw + d(\Gamma)$ and height $2n$ around $v_0$ . Due to  Lemmas.~\ref{lemma:circle_around},~\ref{lemma:n-height} and~\ref{lemma:rpw_dist}, we know that these cylinders are disjoint for different unprocessed vertices that still lie in $XY_0$-plane. Step 10 ensures that Steps 11-13 move vertices only inside this cylinder and the processed part of the tree does not intersect with the unprocessed part.

Steps 11-13 differ depending whether or not we have rotated $T'(v_0)$ during Step 2. In one case $T'(v_0)$ is in $x+$ direction from $v_0$ and in the other $T(v_1)$ is in $x+$ direction from $v_0$. Steps 11 and 13 make two rotations of the needed part of the tree to correct $x, y$-coordinates. Also, we need to move $v_1$ and its subtree to the canonical height with respect to $v_0$. In Step 12, we make the $z$-coordinate correction of $T(v_1)$. The steps are ordered in such a way that no intersections happen during their execution. 

Step 13 concludes the procedure \lift{P_i}, placing all processing vertices into their canonical positions with respect to $v_0$, the head of the lifted path.
\noindent\textit{\textbf{Step 1: Shrink.}} 

\begin{figure}[!htp]

  \begin{center}  
  \includegraphics[scale = 0.6]{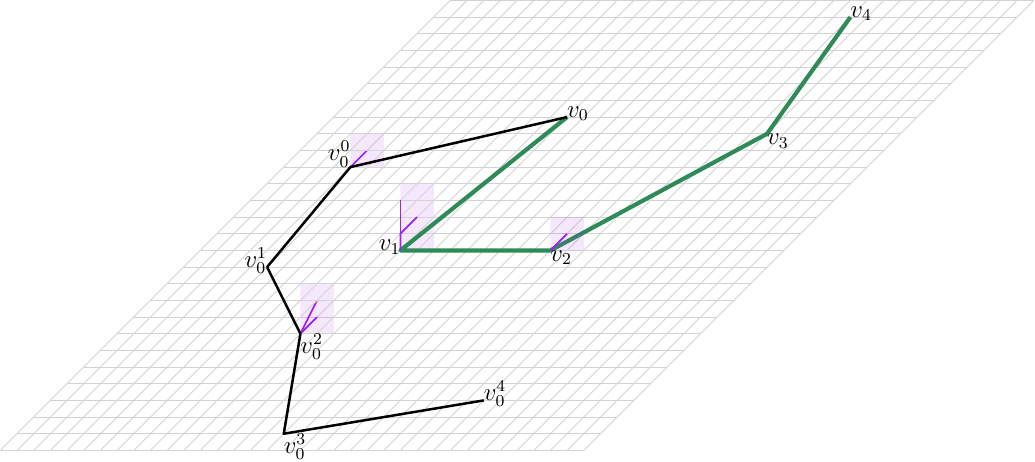}
  \end{center}
  
  \caption{\textbf{Step 1}.}
  \label{fig:step1}

\end{figure}

For every internal vertex $v_j$ of the path $P_i$ its lifted subtree $T'(v_j)$ morphs into the shrunk lifted subtree, see Sec.~\ref{sec:Tools}. All subtrees are shrunk simultaneously in one morphing step. 

\begin{lemma} Step 1 is a crossing-free morph.
\end{lemma}

\begin{proof}
For each lifted subtree this is a shrinking step and by Lemma~\ref{lemma:shrinking} it is crossing-free.

All lifted subtrees were not overlapping in projection to $XY_0$ plane due to Condition (I) and Lemmas~\ref{lemma:circle_around} and~\ref{lemma:rpw_dist} in the drawing $\Gamma_t$. As all vectors of movement are vertical, projections of lifted subtrees can not overlap during this morphing step and can not cross with each other.

In shrunk drawing every vertex of $T'(v_j)$ except $v_j$ has strictly positive $z$-coordinate which means that $T'(v_j) \setminus \{v_j\}$ remains strictly above $XY_0$ plane and can not cross with the part of $T$ that is still lying in the $XY_0$. Internal vertices of $P_i$ do not change their positions during this morph. \qed
\end{proof}

\noindent\textit{\textbf{Step 2: Move others.}} 

\begin{figure}[!htp]

  \begin{center}  
  \includegraphics[scale = 0.6]{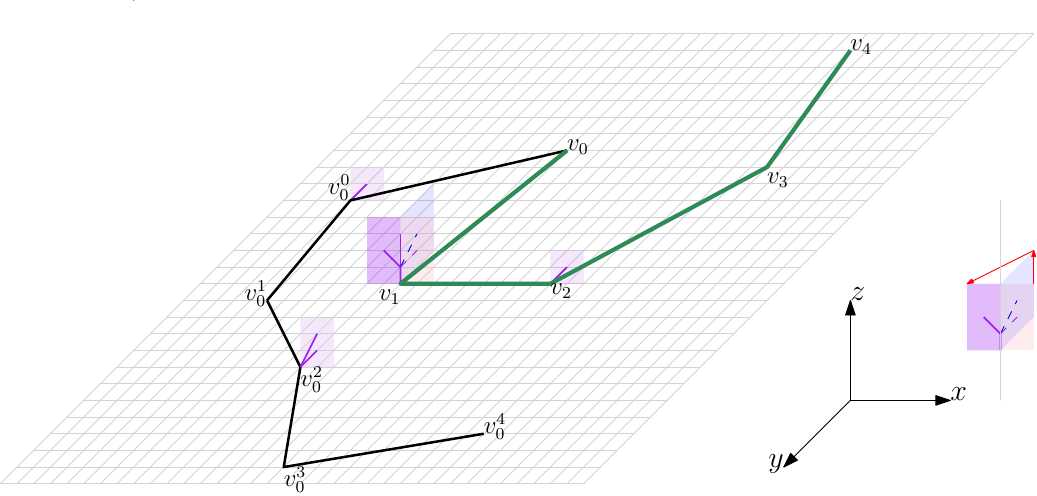}
  \end{center}
  
  \caption{\textbf{Step 2}.}
  \label{fig:step2}

\end{figure}

This step consists of two morphs $\langle \Gamma_{t}, \Gamma_{t + 1} \rangle, \langle \Gamma_{t + 1}, \Gamma_{t + 2} \rangle$.

 For $0 \leq j < i - 1$, if projection $pr(T'(v_j))$ overlaps with $pr((v_j, v_{j + 1}))$, we  rotate twice the drawing of $T'(v_j)$ around the  vertical pole through $\Gamma_t(v_j)$. Since every lifted subtree $T'(v_j)$ lies in $XZ^+_{v_j}$, after this step all lifted subtrees lie in $XZ^+_{v_j}$ or $XZ^-_{v_j}$. 
 
 \begin{lemma} Step 2 is a crossing-free morph. 
\end{lemma}

\begin{proof}
Two different lifted subtrees $T'(v), T'(v')$ cannot cross as due to Lemma~\ref{lemma:rpw_dist}: $pr(T'(v)) \subset \ball{\Gamma_t}{v}{rpw}, pr(T'(v')) \subset \ball{\Gamma_t}{v'}{rpw}$ and due to Lemma~\ref{lemma:circle_around} disks around vertices with radius $rpw \le \frac{\mathcal{S}_1}{
2}$ do not cross with each other.

No edges within rotating lifted subtree can intersect as rotating is a mapping and mapping is a crossing-free morph (Lemma~\ref{lemma:mapping_crossing-free}). Due to condition (II), all edges that do not lie in lifted subtrees are lying in the $XY_0$ plane. No lifted subtree $T'(v_j)$ can intersect with any edges in the $XY_0$ by condition (I): $\Gamma_t(T'(v_j))$ is the canonical drawing with respect to $v_j$ and lies strictly above $XY_0$ plane except for the point $\Gamma_t(v_j)$. Rotation morph does not change $z$-coordinate of points during the movement, so throughout the morph drawing of $T'(v_j)$ lies above $XY_0$ plane, vertex $v_j$ does not move during the rotation as it lies on the pole.

No vertices move in the $XY_0$ plane during Step 2 and no crossing can happen within the plane $XY_0$. \qed
\end{proof}

\noindent\textit{\textbf{Step 3: Go up:}} 

\begin{figure}[!htp]

  \begin{center}  
  \includegraphics[scale = 0.6]{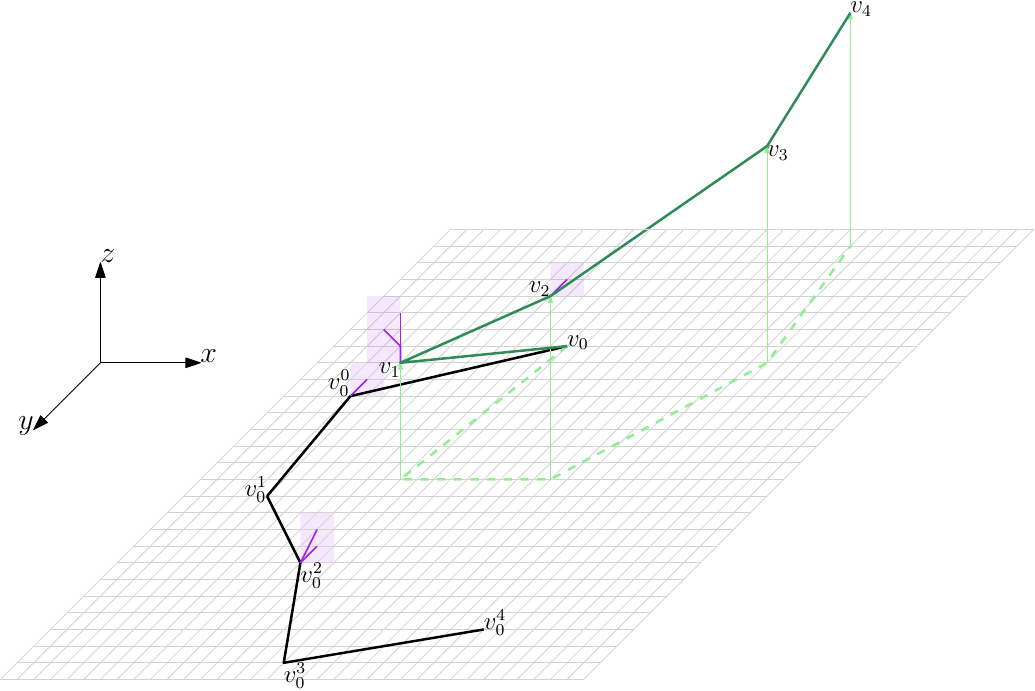}
  \end{center}
  
  \caption{\textbf{Step 3.}}

\end{figure}

This is one morphing step $\langle \Gamma_{t + 2}, \Gamma_{t + 3} \rangle$ that moves each internal vertex $v_j, j \geq 1$ of path $P_i$ vertically to the height defined recursively as follows: 
for $v_1$: $\Gamma_{t + 3}(v_1)_z = n$; for $v_j, j > 1$: $\Gamma_{t + 3}(v_j)_z = \Gamma_{t + 3}(v_{j - 1})_z + h_{sh}(T'(v_{j - 1}))$. Note that $h_{sh}(T'(v_{j}))$ is an integer number equal to the height of the shrunk lifted subtree $T'(v_j)$.

\begin{lemma} 
Step 3 is a crossing-free morph.
\end{lemma}

\begin{proof}
Every internal vertex $v_j$ in the path $P_i$ moves with its subtree $T'(v_j)$ along the vertical vector. In the beginning of this morphing step no intersections existed in projection to the $XY_0$ plane between different subtrees and the path edges. Nothing changes in projection to the $XY_0$ plane during this step, all movements happen strictly above the $XY_0$ plane. That means that there can be no crossings during this step of an algorithm. \qed
\end{proof}

\begin{lemma}
\label{lemma:hor_separ}
\begin{enumerate}
	\item After Step 3,  all internal vertices of $P_i$ along with vertices of their subtrees are horizontally separated from the rest of the vertices of the tree by a horizontal plane. 
	\item The subtrees of internal vertices of $P_i$ are horizontally separated from each other. 
\end{enumerate}
\end{lemma}

\begin{proof}
\begin{enumerate}
	\item Due to Lemma~\ref{lemma:n-height}, $z$-coordinates of all vertices that do not lie on the path are strictly below $n$, also they are integer. By height definition in $\Gamma_{t + 3}$ internal vertices of $P$ and vertices in their lifted subtrees have $z$-coordinates at least $n$.

So the plane $(x, y, n - \frac{1}{2}), x, y \in \mathbb{R}$ is the horizontal plane of separation.
	\item For every pair of internal vertices  $v_j, v_k, j < k$ of $P$,  let us define the plane of separation as follows. From the definition of the height $\Gamma_{t + 3}(v_k)_z$ the horizontal plane $(x, y, \Gamma_{t + 3}(v_k)_z - \frac{1}{2}), x, y \in \mathbb{R}$ separates lifted subtree of vertex $v_j$ from lifted subtree of vertex $v_k$. \qed
\end{enumerate} 
\end{proof}

\noindent\textit{\textbf{Step 4: Rotate subtrees to horizontal plane.}}

\begin{figure}[!htp]

  \begin{center}  
  \includegraphics[scale = 0.6]{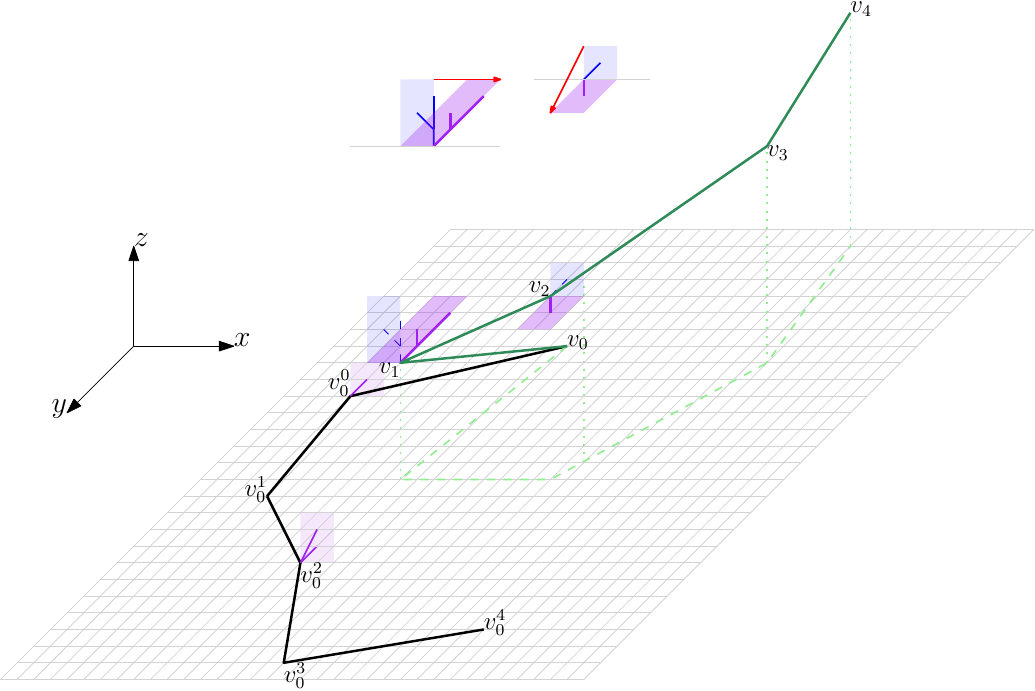}
  \end{center}
  
  \caption{\textbf{Step 4.}}

\end{figure}

In this morphing step, $\langle \Gamma_{t + 3}, \Gamma_{t + 4} \rangle$ all lifted subtrees $T'(v_j)$ of internal vertices of path $P_i$ are rotated around a horizontal pole through $\Gamma_{t + 3}(v_j)$ to lie in a horizontal plane. The direction of rotation is chosen in such a way that $T'(v_j)$ does not cross with an edge $(v_j, v_{j + 1})$ throughout this morph. That is we rotate in the half-space that does not contain the edge $(v_j, v_{j + 1})$.

\begin{lemma} 
With the proper direction of rotation, Step 4 is a crossing-free morph.
\end{lemma}

\begin{proof}
If the crossing happens it must include one of the moving edges, i.e., edges of lifted subtrees. 

As subtrees $T'(v_j)$ are horizontally separated from each other after Step 3 by Lemma~\ref{lemma:hor_separ}, edges of different lifted subtrees can not cross. Also by Lemma~\ref{lemma:hor_separ}, subtrees $T'(v_j)$ can not cross with non-processing vertices or edges between them as they also are horizontally separated in the beginning and in the end of this step.

If for some $j, 1 \le j \le m$ the edge $(v_j, v_{j + 1})$ lied in the vertical plane through $T'(v_j)$ there was no crossings in $\Gamma_{t + 3}$ and after the beginning of movement $T'(v_j)$ will no longer lie in that plane and no crossing can happen. If $(v_j, v_{j + 1})$ lied strictly within one of the half-planes defined by $T'(v_j)$, then by the choice of direction of rotation we rotate through the half-space that does not contain $(v_j, v_{j + 1})$ and thus can not cross with this edge either. \qed
\end{proof}

\noindent\textit{\textbf{Step 5: Correct the path:}}

\begin{figure}[!htp]

  \begin{center}  
  \includegraphics[scale = 0.6]{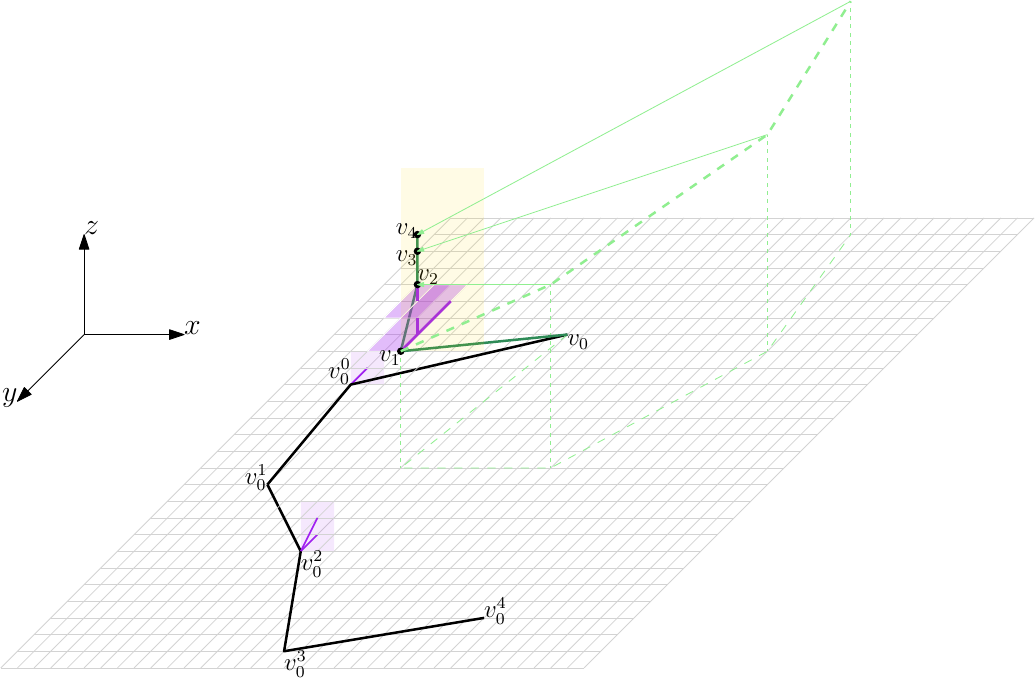}
  \end{center}
  
  \caption{\textbf{Step 5.}}

\end{figure}

In morph $\langle \Gamma_{t + 4}, \Gamma_{t + 5} \rangle$, each vertex $v_j$($j \geq 2$) of the path $P_i$ moves together with its subtree $T'(v_j)$ along the vector $((v_{1_x}  - v_{j_x}) + \mathcal{C}(v_j)_x - \mathcal{C}(v_1)_x, v_{1_y} - v_{j_y}, 0)$, where $v_{1_x}$ denotes $x$-coordinate of vertex $v_1$ in drawing $\Gamma_{t + 4}$. At the end of this step, the $x$ and $y$ coordinates of each $v_j$ are same as their $x$ and $y$ coordinates in the canonical drawing with respect to $v_1$. 

\begin{lemma} Step 5 is a crossing-free morph.
\end{lemma}

\begin{proof}
Step is crossing-free because all edges of the path $i$ and all subtrees of different vertices $v_j$ of the path $P_i$ (that lie in the horizontal planes) are horizontally separated from each other. All moving edges are horizontally separated from already lifted subtrees with the plane $z = n$ by Lemma~\ref{lemma:hor_separ}. \qed
\end{proof}

\noindent\textit{\textbf{Step 6: Go down:}} 

\begin{figure}[!htp]

  \begin{center}  
  \includegraphics[scale = 0.6]{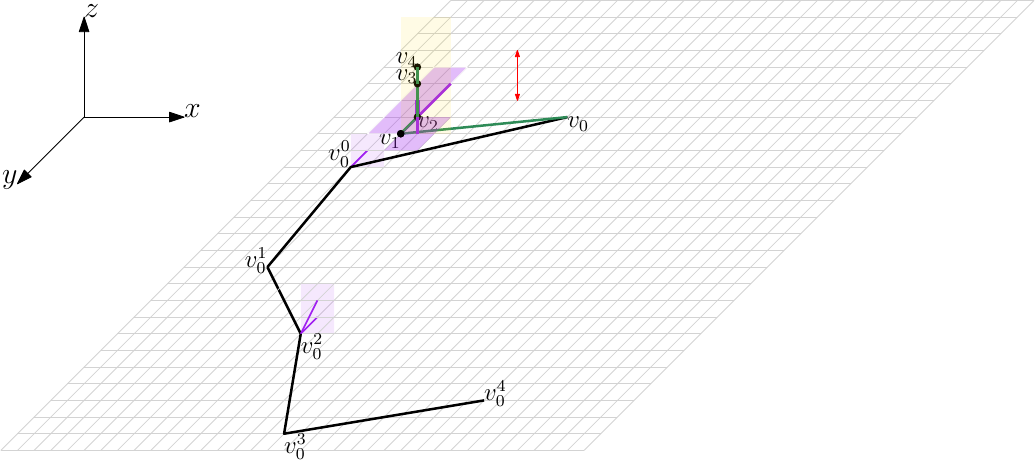}
  \end{center}
  
  \caption{\textbf{Step 6.}}

\end{figure}

During the morphing step $\langle \Gamma_{t + 5}, \Gamma_{t + 6} \rangle$, every vertex $v_j, j \geq 2$ of the path $P_i$ moves together with its subtree $T'(v_j)$ along the same vertical vector $(0, 0, (v_{1_z}  - v_{j_z}) + \mathcal{C}(v_j)_z - \mathcal{C}(v_1)_z)$, where $v_{1_z}$ means $z$-coordinate of vertex $v_1$ in drawing $\Gamma_{t + 5}$. At the end of this step, the $z$ coordinate of $v_j$ is the same as the $z$ coordinate of it in the canonical drawing with respect to $v_1$. 

\begin{lemma} Step 6 is a crossing-free morph.
\end{lemma}

\begin{proof}
All moving edges are horizontally separated from already lifted subtrees with the plane $z = n$. During the morph the vertical order of internal vertices of the path does not change as $P_i$ is a portion of a root-to-leaf path. From that it follows that horizontal separation of $T'(v_j)$ remains for different $j$. \qed
\end{proof}

\noindent\textit{\textbf{Step 7: Turn subtrees in the horizontal planes:}} 

\begin{figure}[!htp]

  \begin{center}  
  \includegraphics[scale = 0.6]{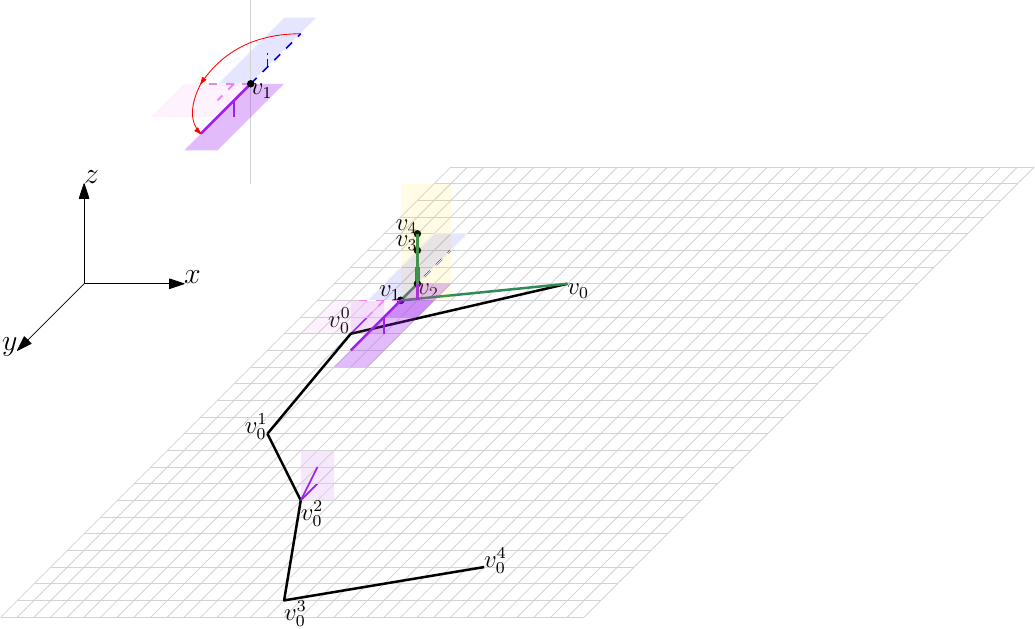}
  \end{center}
  
  \caption{\textbf{Step 7.}}

\end{figure}

It consists of two morphing steps $\langle \Gamma_{t + 6}, \Gamma_{t + 7} \rangle, \langle \Gamma_{t + 7}, \Gamma_{t + 8} \rangle$, and turns every lifted subtree $T'(v_j)$ of internal vertices of $P_i$ to lie in positive $x$-direction with respect to vertex $v_j$. For every subtree that lies in negative $x$-direction in the drawing $\Gamma_{t + 6}$, we use the morphing step from Lemma~\ref{lemma:pinwheel_ref} to obtain a  planar point reflection of $T'(v_j)$ across the point $v_j$ in the horizontal plane containing $T'(v_j)$.

\begin{lemma}
Step 7 is a crossing-free morph.
\end{lemma}

\begin{proof}
The fact that point reflection of a tree across the point of its root can be performed as two steps and those steps are crossing-free follows from Lemma~\ref{lemma:pinwheel_ref}.

Vertices of different lifted subtrees are horizontally separated and separated from non-processing vertices and edges, which means no crossings can happen between them. \qed
\end{proof}

\noindent\textit{\textbf{Step 8: Stretch in $y$-direction:}} 

\begin{figure}[!htp]

  \begin{center}  
  \includegraphics[scale = 0.6]{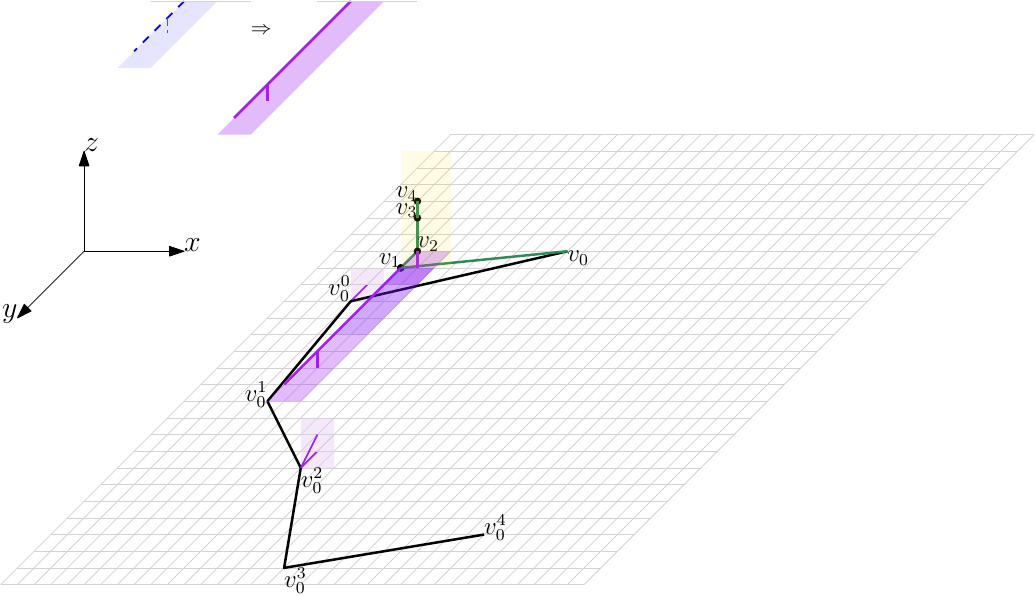}
  \end{center}
  
  \caption{\textbf{Step 8.}}

\end{figure}

This morphing step $\langle \Gamma_{t + 8}, \Gamma_{t + 9} \rangle$ transforms lifted subtrees of internal vertices of $P_i$ in the horizontal planes from shrunk to the canonical size.

\begin{lemma} Step 8 is a crossing-free morph.
\end{lemma}

\begin{proof}
Stretching here is a shrinking morph played backwards and is a crossing-free morph for every $T'(v_j)$ by Lemma~\ref{lemma:shrinking}.

As all $T'(v_j)$ are horizontally separated from each other and from unprocessed part of the tree and morph is horizontal no crossings can happen between different $T'(v_j)$ or $T'(v_j)$ and unprocessed part of the graph. \qed
\end{proof}

\noindent\textit{\textbf{Step 9: Rotate to the canonical positions:}} 

\begin{figure}[!htp]

  \begin{center}  
  \includegraphics[scale = 0.6]{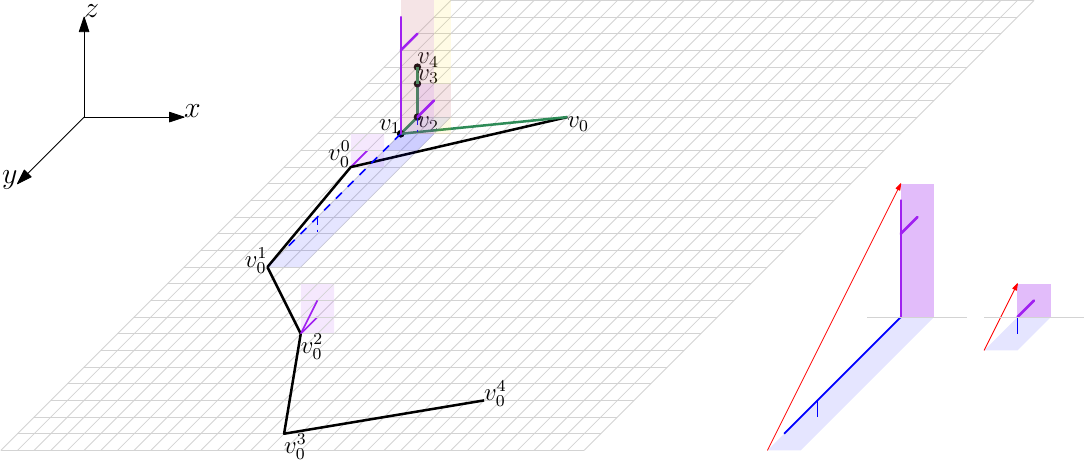}
  \end{center}
  
    \caption{\textbf{Step 9.}}
\end{figure}

In morph $\langle \Gamma_{t + 9}, \Gamma_{t + 10} \rangle$ all lifted subtrees $T'(v_j)$ of the internal vertices of the path rotate around horizontal axes $(x, v_{j_y}, v_{j_z}), x \in \mathbb{R}$ to lie in vertical plane in positive direction such that the subtree $T(v_1)$ is in the canonical position with respect to vertex $v_1$. 

\begin{lemma} Step 9 is a crossing-free morph.
\end{lemma}

\begin{proof}
Let $H$ be the vertical plane parallel to $XZ_0$ and containing $v_1, \ldots, v_m$.

Vertices in subtrees $T'(v_j)$ lie from one side of $H$ in parallel planes at any moment of this morphing step. The half-planes that contain different $T'(v_j)$ are parallel because we can look at the rotation of the half-plane $\alpha$ containing $T'(v_j)$ as at the rotation of the half-plane $\beta$ containing $T'(v_k), j, k \in \{1, \ldots, m\}$ translated by the vector $\Gamma_{t + 9}(v_j)_z - \Gamma_{t + 9}(v_k)_z$. Then at every moment angle $
\angle (\alpha, XZ) = \angle (\beta, XZ)$ which means that these half-planes are parallel to each other. So the vertices of $T'(v_j)$ do not cross. 

$T'(v_j)$ remains in the same half-space from $H$ during all morph for every $j = 1, \ldots, m$. This means that subtrees that lie from different sides of $H$ do not cross too.

At the end of this step subtree of vertex $v_1$ of the drawing $\Gamma_{t + 10}$ is at its canonical positions with respect to to vertex $v_1$ and therefore does not have any crossing within itself. \qed
\end{proof}

\noindent\textit{\textbf{Step 10: Move along.}} 

\begin{figure}[!htp]

  \begin{center}  
  \includegraphics[scale = 0.6]{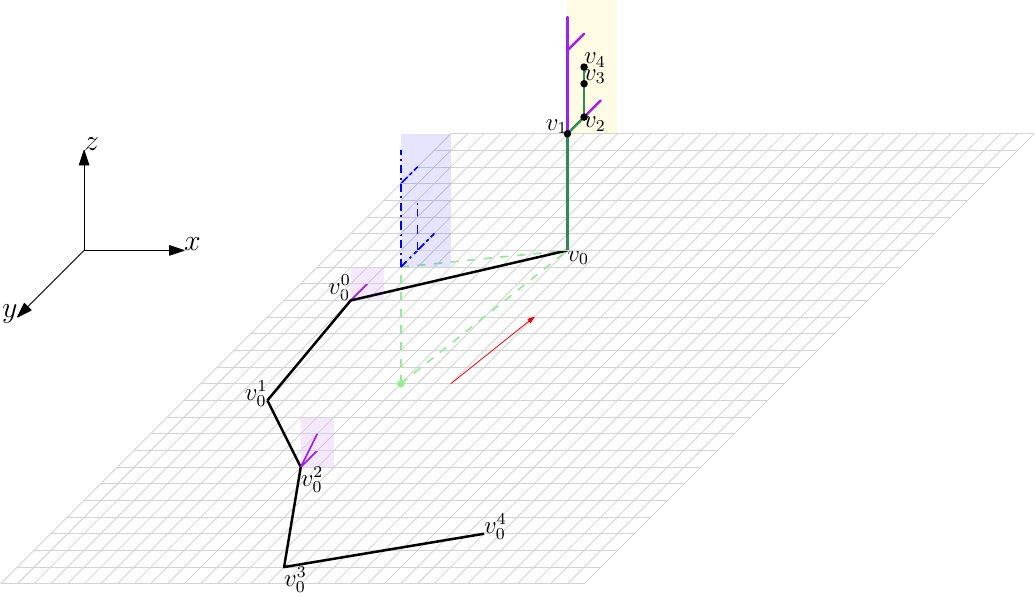}
  \end{center}
  
    \caption{\textbf{Step 10.}}
\end{figure}

In this morphing step, $\langle \Gamma_{t + 10}, \Gamma_{t + 11} \rangle$ every internal vertex $v_j$ of the path with its subtree $T'(v_j)$ moves horizontally in the direction  $(v_{0_x} - v_{1_x}, v_{0_y} - v_{1_y}, 0)$. If in $\mathcal{C}(T)$ the edge $(v_0, v_1)$ is vertical, vertex $v_1$ moves along this vector to get $x, y$-coordinates equal to $(v_{0_x}, v_{0_y})$. Otherwise, vertex $v_1$ moves along this vector as long as possible to get integer $x$ and $y$ coordinates not equal to $(v_{0_x}, v_{0_y})$. From Lemma~\ref{lemma:circle_around} there is such a point $p = (p_x, p_y, 0)$ in disk $\ball{\Gamma_{t + 10}}{v_0}{rpw + d(\Gamma)}$. Note that $\dist{\Gamma_{t + 10}}{p}{v_0} \le d(\Gamma)$. Since all subtrees $T'(v_j)$ take no more than $rpw$ space around the pole $(v_{1_x}, v_{1_y}, z), z \in \mathbb{R}$, the projection to $XY_0$ plane lie within the disk $\ball{\Gamma_{t + 10}}{v_0}{rpw + d(\Gamma)}$. After this step, the projection of all processing vertices lie within $\ball{\Gamma_{t + 10}}{v_0}{rpw + d(\Gamma)}$ in the $XY_0$ plane.

\begin{lemma} Step 10 is a crossing-free morph.
\end{lemma}

\begin{proof}
Step is crossing-free because all the vertices of the path $P$ are horizontally separated from other subtrees of other vertices in the plane $XY_0$ (all already lifted subtrees have height at most $n - 1$ and all processing vertices have $z$-coordinates at least $n$) and the only not-separated edge $(v_0, v_1)$ moves along such a vector, that its projection does not change angle in the $XY_0$ plane. \qed
\end{proof}

All the steps from 11 to 13 are split in two different cases depending on whether we turned $T'(v_0)$ during Step 2 or not. 
\textbf{Case 1} describes all three steps when $T'(v_0)$ was not rotated since $pr(v_0, v_1)$ and $T'(v_0)$ did not overlap after Step 1. This implies that in $\Gamma_{t + 11}$ already lifted subtree of $v_0$ lies in $XZ^+_{v_0}$ (Fig.~\ref{fig:case1}).\textbf{Case 2} describes all three steps if overlap happened and $T'(v_0)$ was rotated twice to lie in $XZ^-_{v_0}$ (Fig.~\ref{fig:step13}).

\noindent\textit{\textbf{Step 11: First part of $xy$-correction:}} \\
\textbf{Case 1:} If in Step 2 $T'(v_0)$ was not rotated, during $\langle \Gamma_{t + 11}, \Gamma_{t + 12} \rangle$ all internal vertices of $P$ with their subtrees $T'(v_j)$ move along the same horizontal vector until the edge $(v_0, v_1)$  lie on the half-plane parallel to $YZ_0$ and $|v_{1_y} - v_{0_y}| = |\mathcal{C}(v_1)_x - \mathcal{C}(v_0)_x|$. The direction is chosen so that the angle between $pr((v_0, v_1))$ in $\Gamma_{t + 11}$ and $pr((v_0, v_1))$ in $\Gamma_{t + 12}$ is minimal.

\noindent\textbf{Case 2:} If during Step 2 $T'(v_0)$ was rotated, which means that the edge $(v_0, v_1)$ was parallel to $0X$, then $pr((v_0 ,v_1))$ is still parallel to $0X$, because we have not changed its direction. By definition of Step 2 we know that $T'(v_0)$ in $\Gamma_{t + 11}$ lies in $XZ^-_{v_0}$. We are rotating $T'(v_0)$ around the pole through $v_0$ to lie in $YZ^+_{v_0}$.  

\begin{lemma} Step 11 is a crossing-free morph.
\end{lemma}

\begin{proof}
\textbf{Case 1:} All vectors of movement in $T(v_1)$ are the same, which means no crossing can happen in $T(v_1)$. Also $T(v_1)$ is horizontally separated from the unprocessed part of $T$, which is motionless. 

\begin{sloppypar}
As for the edge $(v_0, v_1)$ its projection $pr((v_0, v_1))$ lies within the disk ${\ball{\Gamma_{t + 11}}{v_0}{rpw + d(\Gamma)}}$ throughout this morphing step and therefore can not cross with lifted subtrees $T'(v)$ of other non-processing vertices $v$ which projections lie within ${\ball{\Gamma_{t + 11}}{v}{rpw + d(\Gamma)}}$. Also $(v_0, v_1)$ can not cross with $T'(v_0)$ because $(v_0, v_1)$ moves within one half-space defined by plane $XZ_{v_0}$.
\end{sloppypar}

\textbf{Case 2:} Rotation is a crossing-free morph and in projection happens within $\ball{\Gamma_{t + 11}}{v_0}{rpw + d(\Gamma)}$. For the same reasons as in Case 1 no crossings happen.
\end{proof}

\noindent\textit{\textbf{Step 12: Go down.}} 

\begin{figure}[!htp]

  \begin{center}  
  \includegraphics[scale = 0.6]{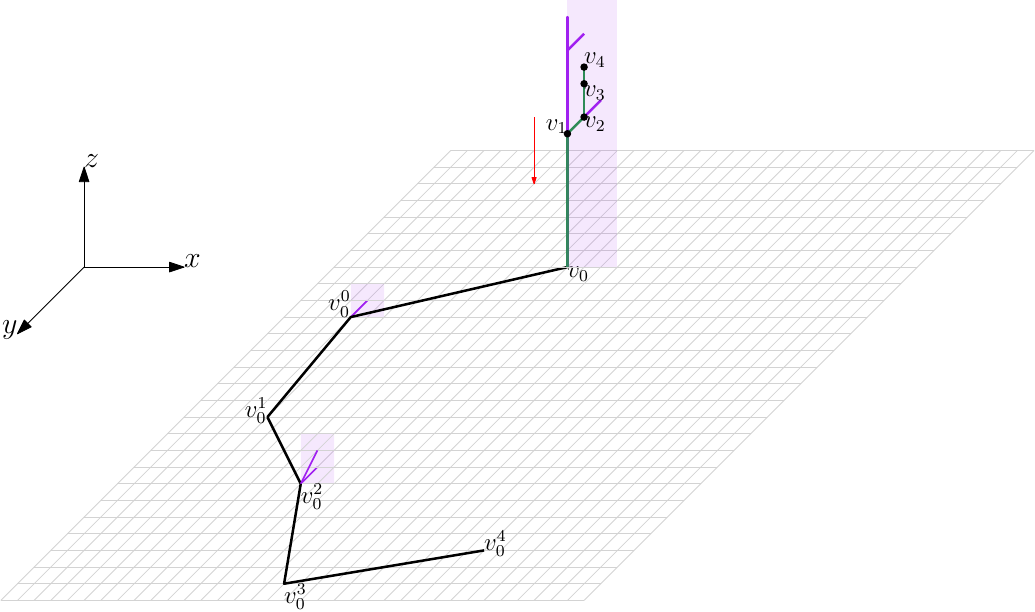}
  \end{center}
  
    \caption{\textbf{Step 12.}}
\end{figure}

$\langle \Gamma_{t + 12}, \Gamma_{t + 13} \rangle$ morphing step is a vertical morph. In $\Gamma_{t + 12}$, the $z$-coordinates of internal vertices of $P$ are $n$ more than their canonical $z$-coordinates with respect to $v_0$. We decrease the $z$-coordinates of internal vertices along with all their subtrees by $n$.

\begin{lemma} Step 12 is a crossing-free morph.
\end{lemma}

\begin{proof}
\textbf{Case 1:} Step is crossing-free because $T'(v_0)$ and $T(v_1)$ lie in distinct parallel planes and do not intersect in projection in $\Gamma_{t + 12}$. Vertical morph does not change the projection of the drawing so separation remains. Edge $(v_0, v_1)$ and $T'(v_0)$ after Step 11 lie in different planes too.

\textbf{Case 2:} During Step 12 $(v_0, v_1)$ and $T(v_1)$ are in $XZ^+_{v_0}$. $T'(v_0)$ during the same step is in $XZ^-_{v_0}$. They can not make any crossings because they do not intersect in projection during all this morph. \qed
\end{proof}

\noindent\textit{\textbf{Step 13: Second part of $xy$-correction.}} This step consists of one morphing step $\langle \Gamma_{t + 13}, \Gamma_{t + 14} \rangle$. 

\textbf{Case 1} All processing vertices move along the same horizontal vector so that after this movement vertex $v_1$ lies in the canonical position with respect to $v_0$. As $T(v_1)$ is already in the canonical position with respect to $v_1$, after Step 13 it is in the canonical position with respect to $v_0$.

\textbf{Case 2} $T(v_1)$ and $(v_0, v_1)$ are in the canonical positions with respect to $v_0$ after Step 12: in Step 10 we got $x, y$-coordinates equal to $(v_{0_x} + (\mathcal{C}(v_1)_x - \mathcal{C}(v_0)_x), v_{0_y} + (\mathcal{C}(v_1)_y - \mathcal{C}(v_0)_y)$ because $(v_0, v_1)$ was parallel to 0X axis, after Step 12 we have corrected $z$-coordinates. In this case, we rotate $T'(v_0)$ to $XZ^+_{v_0}$, i.e. to its canonical position with respect to $v_0$. 
In both cases processing vertices and $T'(v_0)$ lie in the canonical position with respect to $v_0$ in $\Gamma_{t + 14}$, they all now form new lifted subtree of vertex $v_0$.

\begin{figure}[!htp]

  \begin{center}  
  \includegraphics[scale = 0.6]{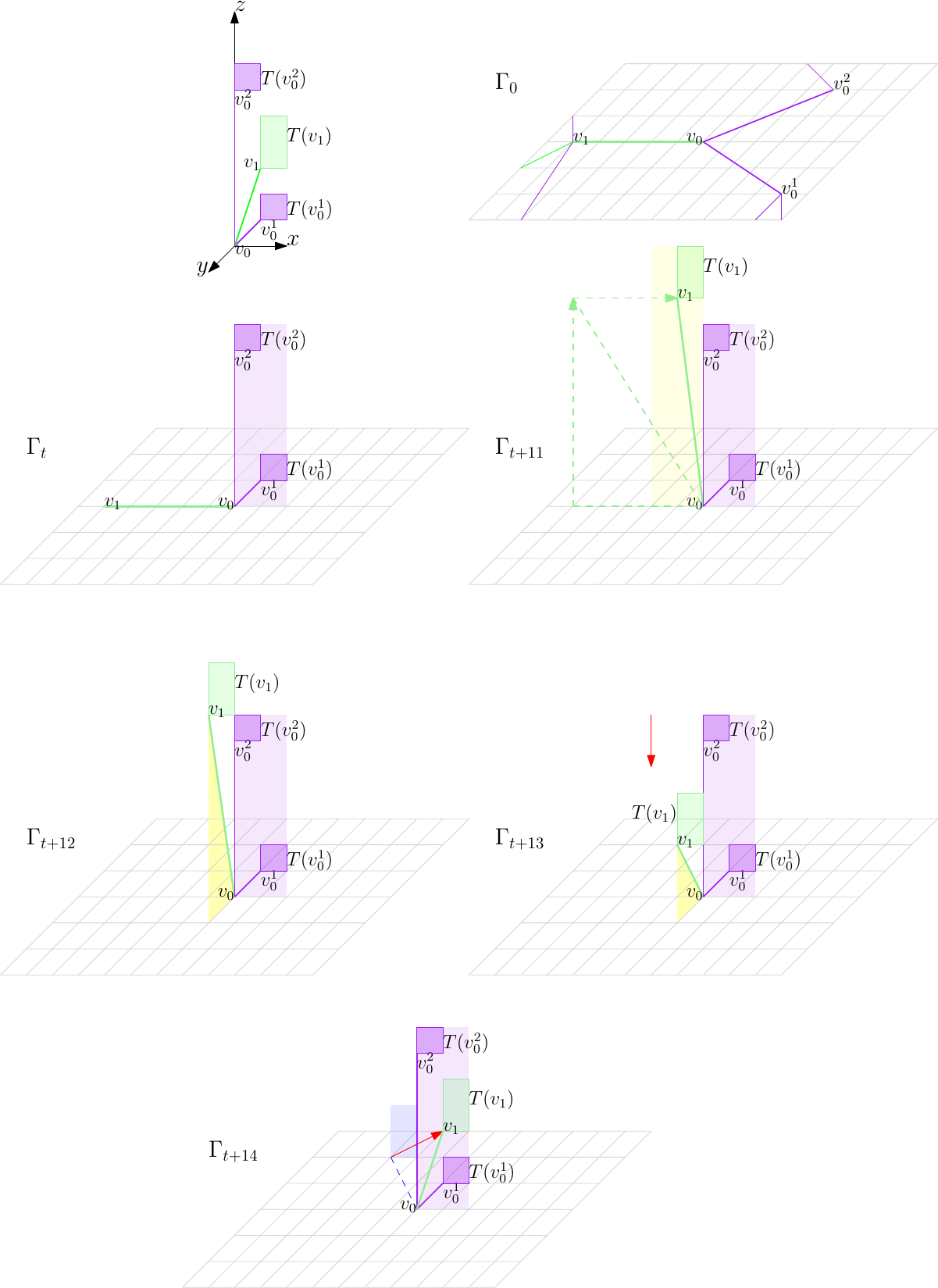}
  \end{center}
  
  \caption{\textbf{Step 11-13.} Case 1}
  \label{fig:case1}

\end{figure}

\begin{figure}[!htp]

  \begin{center}  
  \includegraphics[scale = 0.6]{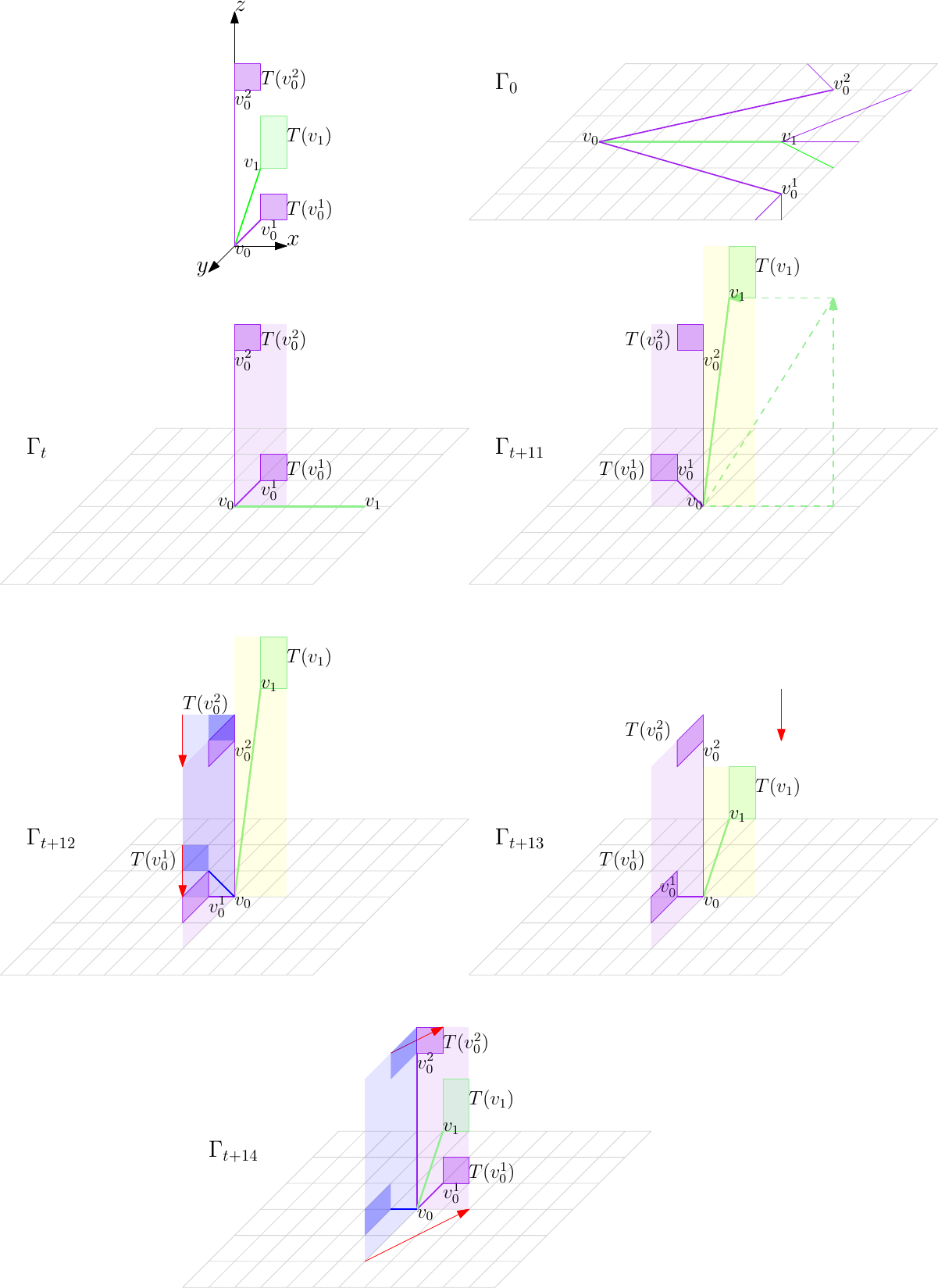}
  \end{center}
  
  \caption{\textbf{Step 11-13.} Case 2}
  \label{fig:step13}

\end{figure}

\begin{lemma} Step 13 is a crossing-free morph.
\end{lemma}

\begin{proof}
\textbf{Case 1:} During the morph projections of $(v_0, v_1)$, $T(v_1)$, $T'(v_0)$ do not cross. In the end of the step $T(v_1)$ and $(v_0, v_1)$ lie in the canonical positions with respect to $v_0$, $T'(v_0)$ was already in the canonical position with respect to $v_0$ by condition (I), so at the end of the morph they can not cross too.

\textbf{Case 2:} Rotation is a crossing-free morph and in the end of this step we get $T'(v_0), T(v_1)$ and $(v_0, v_1)$ to be in the canonical positions with respect to $v_0$ and $\mathcal{C}_{T_{v_0}}$ does not contain any crossings. \qed
\end{proof}

In the end of these morphing steps, we observe that all the internal vertices of $P_i$ along with their subtrees are placed in the canonical position with respect to $v_0$. We keep on lifting up paths until we obtain the canonical drawing of $T$.

\paragraph{\textbf{Correctness of the algorithm}} 

\begin{lemma} Conditions (I) and  (II) hold after performing \lift{P_{i}} for each $1 \le i \le m$.
\end{lemma}

\begin{proof}

\begin{itemize}
	\item[(I)] the drawing of $T'(v)$ in $\Gamma_t$ is the canonical drawing of $T'(v)$ with respect to $v$ for any $v \in V(T)$
	
	\underline{Base of the induction: $i = 1$:}
	
	Note that after performing the stretching step the whole $T$ lies on the $XY_0$ plane. Since none of the paths are lifted, condition (I) trivially holds. 
	
	\underline{The induction step:}	
	
	By induction hypothesis every $T'(v)$ for internal vertex $v$ lies in position needed before \lift{P_i} procedure. After Step 5 all internal vertices have the canonical $x, y$-coordinates with respect to vertex $v_1$ and after Step 6 --- the canonical $z$-coordinates with respect to $v_1$ also. Steps 7 and 8 guarantee that $T'(v_j), j = 1, \ldots, m$ are in the canonical position with respect to to their roots rotated to horizontal plane in positive $x$-direction. Steps 1 and 8 are mutually inverse planar morphs for every $T'(v_j), 1 \le j \le m$ which means that the canonical coordinates with respect to the roots will remain the same after Step 8 but in the horizontal direction. 
	
	Step 9 lifts subtrees of internal vertices into vertical canonical position and $T(v_1)$ is in the canonical position with respect to vertex $v_1$. Steps 10-13 move $T(v_1)$ along with $v_1$, so after \lift{P} procedure $T(v_1)$ is in the canonical position with respect to $v_1$.
	
	Steps 10, 11, 13 place $v_1$ in the canonical $xy$-position with respect to $v_0$ and Step 12 makes it $y$-canonical. So after Step 13 old $T'(v_0)$ along with new edge $(v_0, v_1)$ and subtree $T(v_1)$ is in the canonical position with respect to $v_0$. 
	
	As for the other vertices in the $XY_0$, their lifted subtrees had not moved during \lift{P} procedure and are in the canonical positions with respect to their roots by induction hypothesis.

	\item[(II)] vertices of paths $P_k, k > i$ are lying within the $XY_0$ plane
	
	\underline{Base of the induction: $i = 1$:}	
	
	The condition (II) holds since the entire tree $T$ lies in the $XY_0$ plane.	
	
	\underline{The induction step:}
	
	By induction hypothesis before \lift{P} procedure all vertices of non-processed paths lie in the $XY_0$ plane. Internal vertices of the path $P$ can not lie in the other non-processed paths than $P$ by definition of path decomposition. That means that after \lift{P} in which we move only processed paths, i.e. lifted subtrees, or internal vertices of the $P$, all vertices that lie on non-processed paths will still lie in the $XY_0$ plane.
\end{itemize}  \qed
\end{proof}

\begin{lemma}
All morphing steps in the algorithm are integer.
\end{lemma}

\begin{proof}
Let us prove this by induction on number of lifted paths.

The base case is trivial. After the stretching morph all coordinates of all vertices are integer because the constant of stretching is integer and in the given drawing $\Gamma$ all vertices had integer coordinates.

By induction hypothesis in the beginning of \lift{P} procedure all vertices lie on lattice points of the grid. During \lift{P} procedure non-processed  vertices and vertex $v_0$ do not change coordinates at all. Rotation and shrinking morphs move points with integer coordinates to points with integer coordinates. Turning of the subtrees in the horizontal planes in Step 7 is integer by definition (see~\cite{bblbfpr19}). In all other steps the lifted subtrees of internal vertices are moved along with their roots by the integer vector.

As all coordinates at the beginning of \lift{P} were integer and all vectors of movement of all vertices in every step were integer, after \lift{P} procedure all vertices still have integer coordinates. \qed
\end{proof}

\paragraph{\textbf{Complexity of the algorithm}}

Let us estimate the size of the required grid:

Our graph $T = (V, E)$, $|V| = n$ will have drawing $\Gamma = \Gamma_0$ that takes space $l(\Gamma) \times w(\Gamma) \times 1$. 

First step of the algorithm multiplies needed space by $\mathcal{S}_1 = 2 \cdot (rpw + d(\Gamma))$. In the beginning of our algorithm, lifted subtrees of all vertices lie in planes parallel to $XZ_0$ plane passing through the corresponding vertices. They take no more than $\mathcal{O}(rpw)$ space in $x, y$ direction and no more than $n-1$ in $z$ direction. During the execution of the algorithm (during Step $4$-Step $8$), we rotate lifted subtrees of the internal vertices of the path to lie in horizontal planes passing through the corresponding vertices, at this point the subtrees take no more than $n - 1$ space in $y$-direction (and same $\mathcal{O}(rpw)$ in $x$-direction). During \enquote{Pinwheel} rotation in Step $7$, those subtrees take no more than $n - 1$ space in $x$-direction (and same $\mathcal{O}(rpw)$ in $y$-direction).  As in $z$-direction every lifted subtree drawing takes no more than $n$ height, during \lift{} procedure, we may get vertices at height at most $n + n$.

So the space is:
 $$(x \times y \times z): \; \mathcal{O}((l(\Gamma) \cdot 2 \cdot (d(\Gamma) + rpw) + 2 \cdot n) \times (w(\Gamma) \cdot 2 \cdot (d(\Gamma) + rpw) + 2 \cdot n) \times n) = $$
 $$ \mathcal{O}(d^2(\Gamma) \times d^2(\Gamma) \times n)$$
 
Note that the estimation of $rpw(T) = \mathcal{O}(\log n)$ ~\cite{PoleDance} and it is asymptotically less than $d(\Gamma) \geq  \sqrt{n}$. Also note that $d^2(\Gamma) \geq n$. This implies that  we can accommodate horizontal  subtrees in the grid of aforementioned size. 
The lemmas proved in this section along with the space and time complexity bounds prove the following Theorem.
\begin{theorem} 
\label{thm:sec4}
For every two planar straight-line grid drawings $\Gamma, \Gamma'$ of tree $T$ with $n$ vertices there exists a crossing-free 3D-morph $\mathcal{M} = \langle \Gamma = \Gamma_0, \ldots, \Gamma_l = \Gamma' \rangle$ that takes $\mathcal{O}(k)$ steps where $k$ is number of paths in some path decomposition of tree $T$. In this morph, every intermediate drawing $\Gamma_i, 1 \le i \le l$ is a straight-line 3D grid drawing lying in a grid of size $\mathcal{O}(d^2 \times d^2 \times n)$, where $d$ is maximum of the diameters of the given drawings. 
\end{theorem}

\begin{corollary}
\label{theorem:1alg}
For every two planar straight-line grid drawings $\Gamma, \Gamma'$ of tree $T$ with $n$ vertices there exists a crossing-free 3D-morph $\mathcal{M} = \langle \Gamma = \Gamma_0, \ldots, \Gamma_l = \Gamma' \rangle$ that takes $\mathcal{O}(n)$ steps and $\mathcal{O}(d^2 \times d^2 \times n)$ space to perform, where $d$ is maximum of the diameters of the given drawings. In this morph, every intermediate drawing $\Gamma_i, 1 \le i \le l$ is a straight-line 3D grid drawing. 
\end{corollary}

\begin{proof}
Bound $\mathcal{O}(n)$ to the number of paths in $\mathcal{P}$ is obvious.
\end{proof}

\section{Morphing through lifting edges}
\label{sec:second_alg}


In this section, we describe another algorithm that morphs a planar drawing $\Gamma$ of tree $T$ to the canonical drawing $\mathcal{C}(T)$ of $T$. 
This time one iteration of our algorithm lifts simultaneously a set of edges with at most one edge of each path of a selected path decomposition.
Let $\Gamma = \Gamma_0$ be a planar drawing of  $T$.

\noindent\textit{\textbf{Step 0: Preprocessing.}} This step $\langle \Gamma, \Gamma_1 \rangle$ is a stretching morph with $\mathcal{S}_1 = 2 \cdot rpw \cdot d(\Gamma) \cdot (4 \cdot d(\Gamma) + 1)$, stretching is a crossing-free morph.

\subsection{\liftk{\textit{edges}} procedure}
For edge $e$ of $T$, let \textit{starting vertex of the edge} $st(e)$ (respectively, \textit{ending vertex of the edge} $end(e)$) be the vertex of  $e$ with smallest (respectively, largest) depth. Let $\mathcal{K} = \{K_1, \ldots, K_m\}$ be the partition of edges of $T$ into disjoint sets such that $e \in K_i$ if and only if $ dpt(st(e)) = m - i$, where $m$ denotes the depth of $T$. 
We lift up sets $K_i$ from $\mathcal{K}$ from $i = 1$ to $i = m$ by running the following  procedure \liftk{K_i} (Steps 1-5). Let $\Gamma_t$ be the drawing of $T$ before lifting set $K_i$. \par

Let \textit{lifted subtree} $T'(v_j)$ be the portion of subtree $T(v_j)$ lifted by the execution of \liftk{K_j} where $j < i$. Conditions that we maintain throughout the algorithm are the following: 

\begin{itemize}
	\item[(I)] The drawing of $T'(v)$ in $\Gamma_t$ is the canonical drawing of $T'(v)$ with respect to $v$.
	\item[(II)] The vertices that lie on edges in a set that is not yet processed are lying in the $XY_0$ plane.
\end{itemize} 

Let the \textit{processing} vertices be the ending vertices $end(e)$ of every edge $e$ from $K_i$ along with the vertices of their subtrees $T(end(e))$.

\begin{figure}[!htp]

  \begin{center}  
  \includegraphics[scale = 0.5]{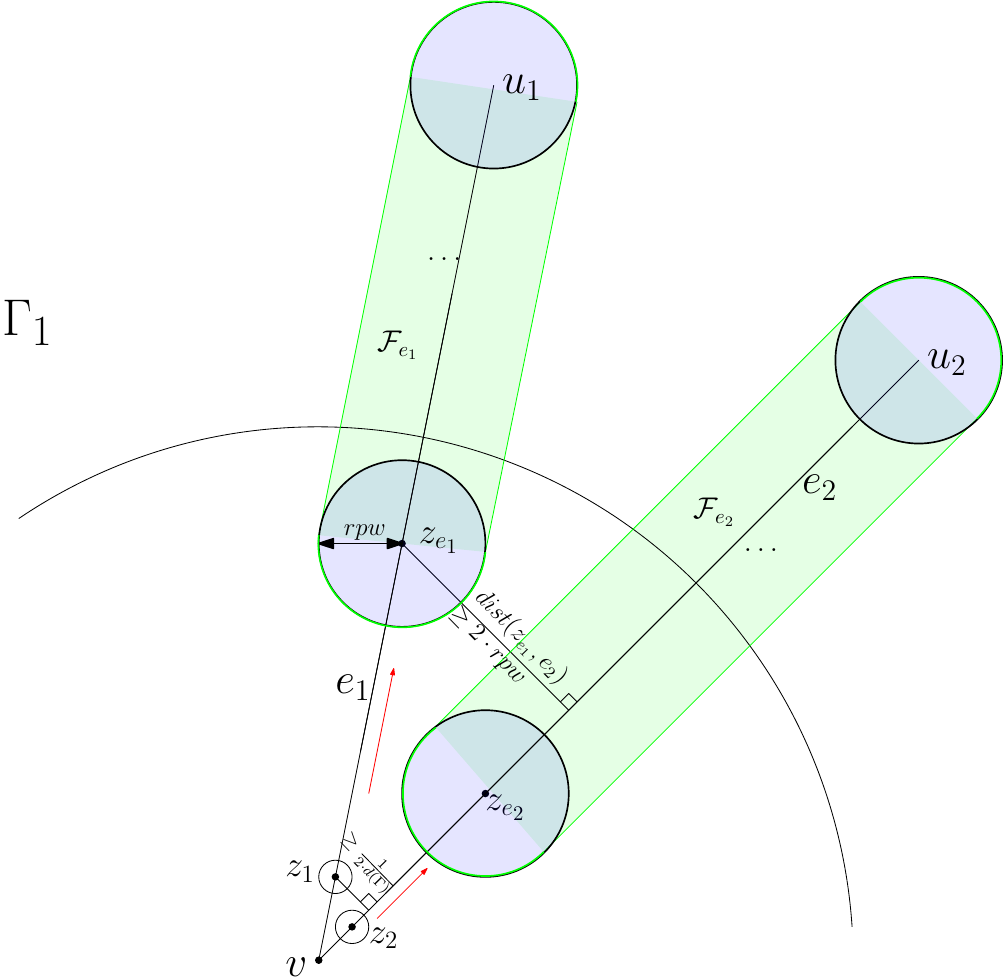}
  \end{center}
  
  \caption{Finding points $z_{e_1}, z_{e_2}$ for edges $e_1, e_2$ with common starting point $v$.}
  \label{fig:point_areas_in_circle}
  
\end{figure}

\begin{lemma} 
\label{lemma:points_in_circles}
For every edge $e = (v, u)$ with $st(e)=v$ in $\Gamma_1$ there is a lattice point $z_e \in e$ such that: 
\begin{enumerate}
	\item $\ball{\Gamma_1}{z_e} {rpw \cdot d(\Gamma)} \subset \ball{\Gamma_1}{v}{rpw \cdot d(\Gamma) \cdot (4 \cdot d(\Gamma) + 1)}$
	\item for distinct edges $e_1, e_2 \in K_i \forall i = 1, \ldots, m$ disks $\ball{\Gamma_1}{z_{e_1}}{rpw}$ and $\ball{\Gamma_1}{z_{e_2}}{rpw}$ are disjoint.
	\item for distinct edges $e_1, e_2 \in K_i \forall i = 1, \ldots, m$ regions $\mathcal{F}_{e_1}, \mathcal{F}_{e_2}$ are disjoint, where 
	$\mathcal{F}_{e} = \{x \in XY_0: dist_{\Gamma_1}(x, (z_{e}, u)) \le rpw\}$.
\end{enumerate}
\end{lemma}
\begin{proof}
 Let us fix an edge $e = (v, u)$. By Lemma~\ref{lemma:circle_around}, there is a lattice point $z$ lying on $e$ in $\ball{\Gamma_1}{v}{d(\Gamma)}$. Then let point $z_e$ be $(v_x + (z_x - v_x) \cdot 4 \cdot rpw \cdot d(\Gamma), v_y + (z_y - v_y) \cdot 4 \cdot rpw \cdot d(\Gamma), 0)$. It satisfies the conditions of the lemma:

\begin{enumerate}
	\item $\ball{\Gamma_1}{z_e}{rpw \cdot d(\Gamma)} \subset \ball{\Gamma_1}{v}{rpw \cdot d(\Gamma) \cdot (4 \cdot d(\Gamma) + 1)}$ holds because $\dist{\Gamma_1}{z_e}{v} \leq d(\Gamma) \cdot (4 \cdot rpw \cdot d(\Gamma))$.

	\item For edges $e_1, e_2 \in K_i$ with different start vertices $v_1, v_2$ we get disjointedness of the corresponding disks from the fact that $\ball{\Gamma_1}{z_{e_j}}{rpw} \subset \ball{\Gamma_1}{v_j}{rpw \cdot d(\Gamma) \cdot (4 \cdot d(\Gamma) + 1)}$ and disks $\ball{\Gamma_1}{v}{rpw \cdot d(\Gamma) \cdot (4 \cdot d(\Gamma) + 1)}$ do not cross for different $v$. 
	
	For edges $e_1, e_2 \in K_i$ with a common start vertex $v$ we get points $z_1, z_2$ from Lemma~\ref{lemma:circle_around} in $\Gamma_1$. We have $dist(z_1, z_2) \geq 1$ because $z_1, z_2$ are lattice points of the grid. Then by definition of $z_{e_1}, z_{e_2}$ we have $dist_{\Gamma_1}(z_{e_1}, z_{e_2}) \geq 2 \cdot rpw$ and for different edges $e$ in $K_i$ disks $\ball{\Gamma_1}{z_e}{rpw}$ do not intersect.

	\item For $e_1 = (v, u_1), e_2 = (v, u_2)$ by Lemma~\ref{lemma:pick} $dist(z_1, e_2) \geq \frac{1}{2 \cdot d(\Gamma)}$, because $z_1, (v, z_2)$ can be interpreted as drawing within \ball{\Gamma_1}{v}{d(\Gamma)} with the diameter at most $2 \cdot d(\Gamma)$. Then $dist(z_{e_1}, e_2) \geq (4 \cdot rpw \cdot d(\Gamma)) \cdot \frac{1}{2 \cdot d(\Gamma)} = 2 \cdot rpw$. 
	
	Minimum distance between $(z_{e_1}, u_1), (z_{e_2}, u_2)$ is realized at one of the endpoints of the segments, for $z_{e_j}, j = 1, 2$ it is at least $2 \cdot rpw$ as we know from above. For vertices $u_j, j = 1, 2$ by Lemma~\ref{lemma:circle_around} disk \ball{\Gamma_t}{u_j}{2 \cdot rpw} does not intersect with any edges non-incident to $u_j$ or contain any other vertices than $u_j$. That means that all segment $(z_{e_1}, u_1)$ is at distance at least $2 \cdot rpw$ from segment $(z_{e_2}, u_2)$ and $\mathcal{F}_{e_1}, \mathcal{F}_{e_2}$ do not intersect.
	 
	 For edges $e_1, e_2$ with different start vertices we have that $e_1, e_2$ do not have common endpoints by the definition of $K_i$. Then $e_1 = (v_1, v_2), e_2 = (v_3, v_4)$ and for the same reasons as above the corresponding regions are disjoint. \qed
\end{enumerate}

\end{proof}

\begin{figure}[!htp]

  \begin{center}  
  \includegraphics[scale = 0.6]{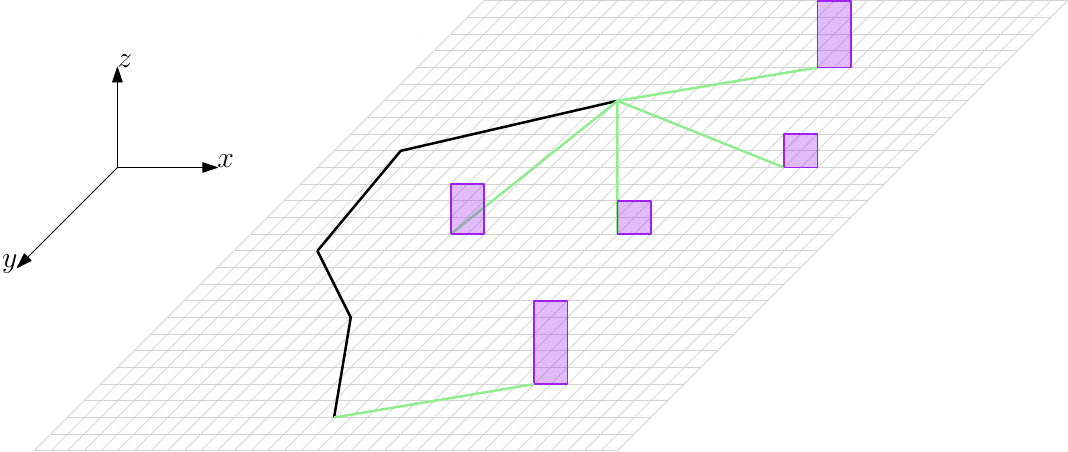}
  \end{center}

  \caption{\textbf{Drawing $\Gamma_t$} The drawing $\Gamma_t$ in the beginning of the procedure \liftk{K_i}, lifted subtrees are violet. $K_i$ consists of green edges.}
\end{figure}

\paragraph{\textbf{Overview}}

The procedure \liftk{K_i} consists of 5 steps and results in moving the end vertices of the edges of the set $K_i$ along with their subtrees rooted at them to their canonical positions with respect to the corresponding starting vertices of the edges of $K_i$. Note that by preprocessing Step 0, Lemma~\ref{lemma:circle_around} and by the similar arguments as in Lemmas~\ref{lemma:n-height} and~\ref{lemma:rpw_dist}, we ensure that  already lifted subtrees lie in disjoint cylinders of radius $rpw$ and height $n$.

During Step 1 and 2 vertices of each processing subtree $T(end(e))), e \in K_i$, move along the same vector as the root $end(e)$. In Step 1 we move simultaneously ending vertices of the edges horizontally towards the starting vertices of the corresponding edges in $K_i$ until they reach certain points specified in Lemma~\ref{lemma:points_in_circles}. Firstly, this step ensures that all later steps are performed in different disjoint cylinders for different starting vertices of the edges in $K_i$. Secondly, it also ensures that for the ending vertices with a common starting vertex, Step 3 is a crossing free morph.

Step 2 moves all ending vertices to their canonical height with respect to the corresponding starting vertices in $K_i$. This step corrects $z$-coordinates of all the processing vertices with respect to the corresponding  starting vertices. 

The following Steps 3-5 move the processing vertices only horizontally and correct their $x, y$-coordinates. Step 3 simultaneously maps every lifted subtree $T(end(e))$ to the vertical plane containing the corresponding edge $e \in K_i$. Idea of this step is to put different subtrees $T(end(e))$ with a common vertex $st(e)$ to separate vertical half-planes around the pole through $st(e)$ and to avoid intersections during Step 5.

Step 4 moves the ending vertices of edges in $K_i$ towards the pole trough the corresponding starting vertices until the last integer point. After this step the drawing in every half-plane that contains edge $e$ and subtree $T(end(e))$ differs from the canonical drawing of this part of the tree with respect to the $st(e)$ by the rotation and by some stretching factor.

Step 5 collapses all planes around each of the starting vertices in $K_i$ in one. In every iteration it divides neighbouring planes in pairs and maps one to another. Step 5 may take up to $\log(\Delta(T))$ morphing steps, where $\Delta(T)$ is the maximum degree of the vertices in $T$.  

Step 5 concludes the procedure as after the collapse of all half-planes around each vertex $st(e)$ and mapping it to the $x+$ direction from $st(e)$ we get the canonical drawing of $T(st(e))$ with respect to $st(e)$.

\noindent\textit{\textbf{Step 1: Shrink.}} 

\begin{figure}[!htp]

  \begin{center}  
  \includegraphics[scale = 0.6]{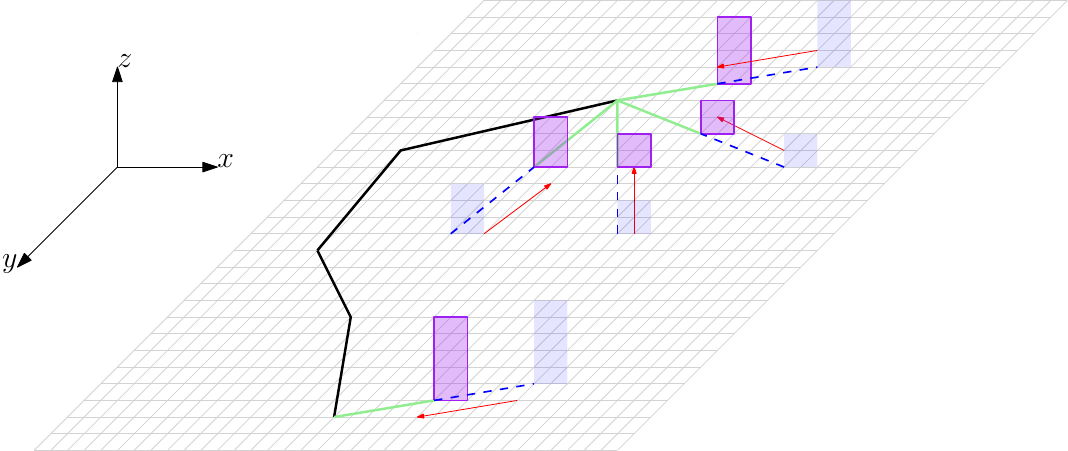}
  \end{center}

  \caption{\textbf{Step 1.}}
\end{figure}

In the  step $\langle \Gamma_{t}, \Gamma_{t + 1}\rangle$, for every edge $e \in K_i$ we move vertex $end(e)$ along with its lifted subtree towards $st(e)$ until $end(e)$  reaches point $z_e$. 

\begin{lemma}\label{lem:sec5step1}
Step 1 is a crossing-free morph.
\end{lemma}

\begin{proof}
Vertices that lie in the $XY_0$ plane do not move or move along their incident edges, so no intersections can happen in the $XY_0$.

For the moving subtrees we know from Lemma~\ref{lemma:points_in_circles} that they do not intersect because each lifted subtree in projection lies in the corresponding region $\mathcal{F}_e$ throughout the whole morph and these regions do not intersect for different $e \in K_i$.  

From Conditions (I) and (II) it follows that no intersections happen during this morphing step. \qed
\end{proof}

\noindent\textit{\textbf{Step 2: Go up.}} 

\begin{figure}[!htp]

  \begin{center}  
  \includegraphics[scale = 0.6]{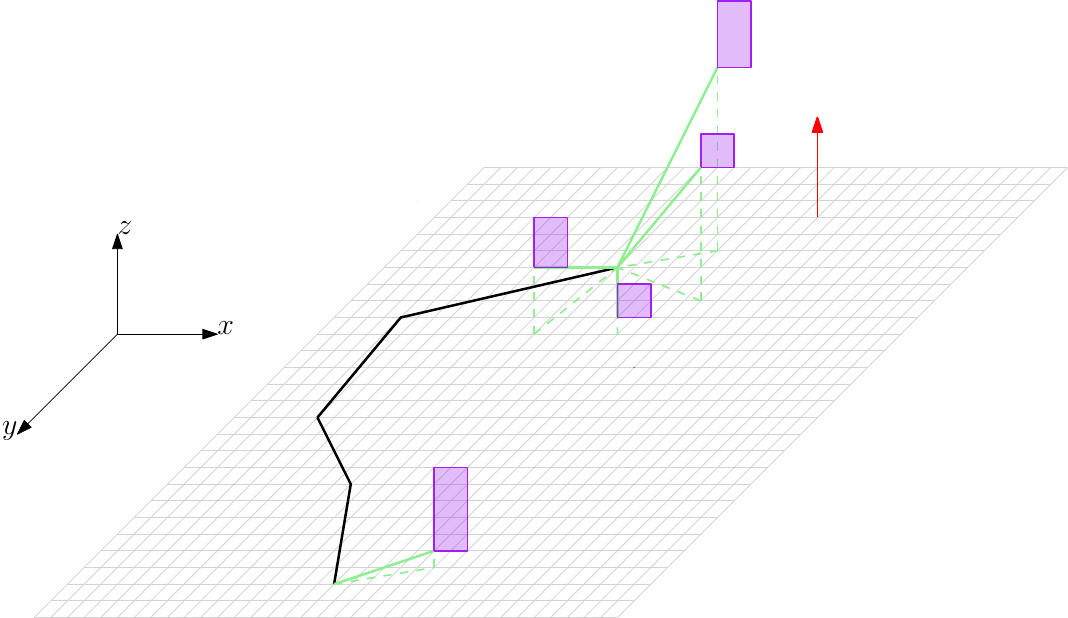}
  \end{center}

  \caption{\textbf{Step 2.}}
\end{figure}

In morphing step $\langle \Gamma_{t + 1}, \Gamma_{t + 2} \rangle$, we move  $end(e)$ with $T'(end(e))$ along the vector $(0, 0, \mathcal{C}(end(e))_z - \mathcal{C}(st(e))_z)$ for all $e \in K_i$. 

\begin{lemma}
Step 2 is a crossing-free morph.
\end{lemma}

\begin{proof}
For every vertex $v$ let $End(v)$ be a set of edges, for which $v$ is a start vertex. All edges of $End(v)$ are contained in one set $K_i$ because their depth is defined by $dpt(v)$. That means that if we are lifting some edges with start vertex $v$, then $v$ does not have any lifted subtree in the beginning of \liftk{K_i} procedure. In projection to $XY_0$ plane in $\Gamma_{t + 1}$ lifted subtrees of the same or different vertices do not intersect.

During this morphing step we do not change projection and move every lifted subtree by the same vertical vector, so no intersections can happen between the subtrees and between the edges of the same subtree. Vertices of unprocessed sets are lying still in the $XY_0$ plane and can not make any intersections too. \qed
\end{proof}

\noindent\textit{\textbf{Step 3: Mapping.}} 

\begin{figure}[!htp]

  \begin{center}  
  \includegraphics[scale = 0.6]{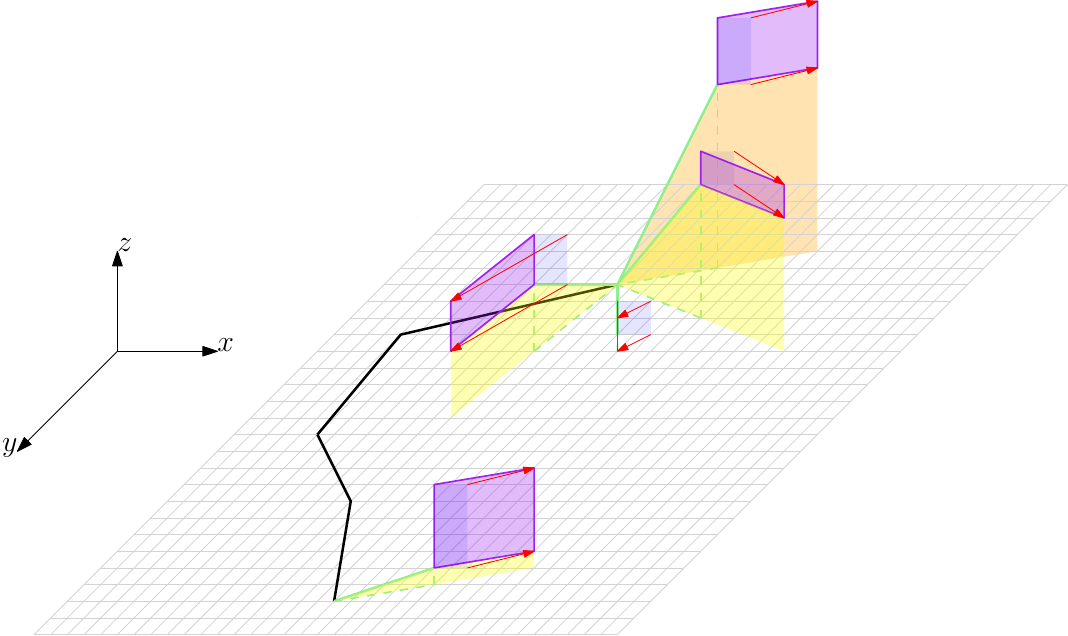}
  \end{center}

  \caption{\textbf{Step 3.}}
\end{figure}

Morphing step $\langle \Gamma_{t + 2}, \Gamma_{t + 3} \rangle$ is a mapping morph, see Section~\ref{sec:mapping_pole}. For every lifted subtree $T'(v_j)$, where $v_j = end(e), e \in K_i$, we define the half-planes of the mapping morph as follows: half-plane $\alpha$ is $XZ^+_{v_j}$, half-plane $\beta$ is part of the vertical plane containing the edge $e$ in such direction that $e \notin \beta$, the common vertical pole of $\alpha$ and $\beta$ is a pole through $v_j$. All mapping steps are done simultaneously for all subtrees of end vertices of the edges of $K_i$. 

\begin{lemma}
Step 3 is a crossing-free morph.
\end{lemma}

\begin{proof}
By Lemma.~\ref{lemma:mapping_crossing-free} mapping is a crossing-free morph and no intersections happen in every $T'(v_j)$. 

Movement of every $T'(v_j)$ for $v_j = end(e)$ happens in projection to $XY_0$ plane in the region $\mathcal{F}_e$ defined for $e$ and $st(e)$: distance from vertices of $T'(v_j)$ to $e$ decreases and $e$ has at least $rpw + 1$ integer points on it so in $\Gamma_{t + 1}$ and  $\Gamma_{t + 2}$ projections of vertices of $T'(v_j)$ lie in $\mathcal{F}_e$.

Different lifted subtrees do not intersect as they do not intersect in projection to $XY_0$ during this step. Other vertices do not move and also make no intersections. \qed
\end{proof}

\noindent\textit{\textbf{Step 4: Shrink more.}}

\begin{figure}[!htp]

  \begin{center}  
  \includegraphics[scale = 0.6]{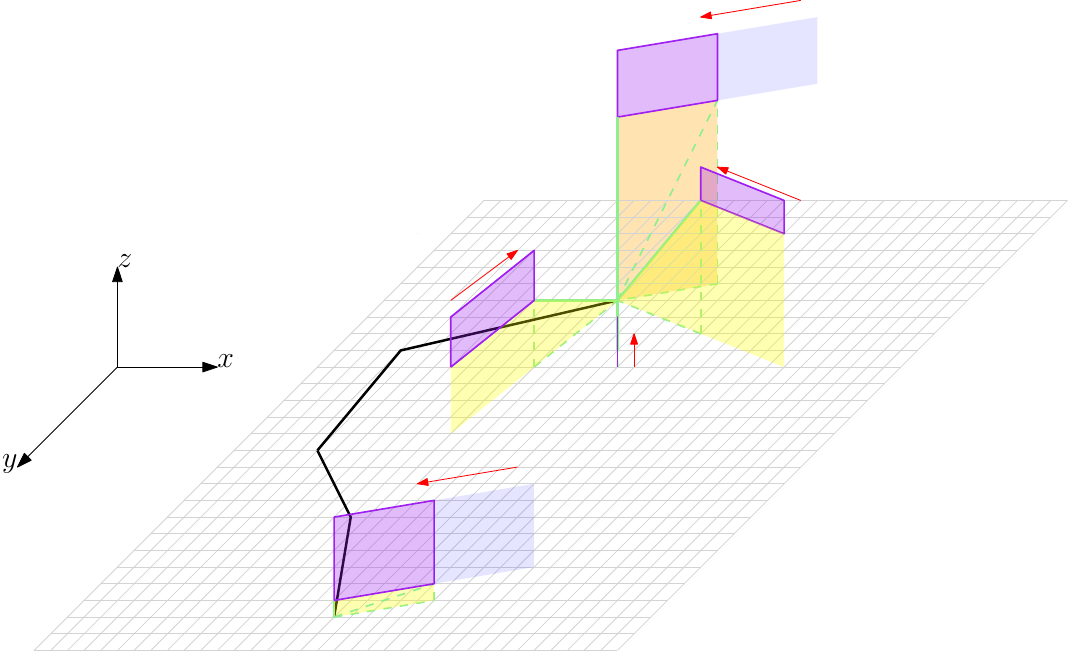}
  \end{center}

  \caption{\textbf{Step 4.}}
\end{figure}

The morphing step $\langle \Gamma_{t + 3}, \Gamma_{t + 4} \rangle$ is a horizontal morph. 
\begin{sloppypar}
For each $v_j = end(e), e \in K_i$ we define a horizontal vector of movement as following. If $e$ is a vertical edge in the canonical drawing then this vector is $(\Gamma_{t + 3}(st(e))_x - \Gamma_{t + 3}(end(e))_x, \Gamma_{t + 3}(st(e))_y - \Gamma_{t + 3}(end(e))_y, 0)$, in this case subtree $T'(end(e))$ is moving towards vertical pole through $st(e)$ until the image of the edge $e$ becomes vertical.

If $e$ is not a vertical edge in the canonical drawing, then $\mathcal{C}(end(e))_x - \mathcal{C}(st(e))_x = 1$ and we move the whole subtree $T'(end(e))$ in the same direction towards the pole through $\Gamma_{t + 3}(st(e))$ until $end(e)$ reaches the last point with integer coordinates before $(\Gamma_{t + 3}(st(e))_x, \Gamma_{t + 3}(st(e))_y, \Gamma_{t + 3}(end(e))_z)$. 
Note that $end(e)$ maps to the same point if we map the canonical drawing of $end(e)$ with respect to the $st(e)$ from $XZ^+_{st(e)}$ to the vertical plane containing $end(e)$.
\end{sloppypar}

\begin{lemma}
Step 4 is a crossing-free morph.
\end{lemma}

\begin{proof}
Every lifted subtree moves inside its half-plane. That means that different subtrees $T'(end(e_j))$ with the same $st(e_j)$ do not intersect in projection to $XY_0$ plane and though do not intersect with each other.

Subtrees with different $st(e_j)$ in projection lie in non-crossing disks $\ball{\Gamma_1}{v}{rpw \cdot d(\Gamma) \cdot (4 \cdot d(\Gamma) + 1)}$ by Lemma~\ref{lemma:points_in_circles} and also can not intersect during this morph. \qed
\end{proof}

\sloppypar{
\noindent\textit{\textbf{Step 5: Collapse planes.}} 

\begin{figure}[!htp]
\centering
\subfloat[]{\includegraphics[scale = 0.35]{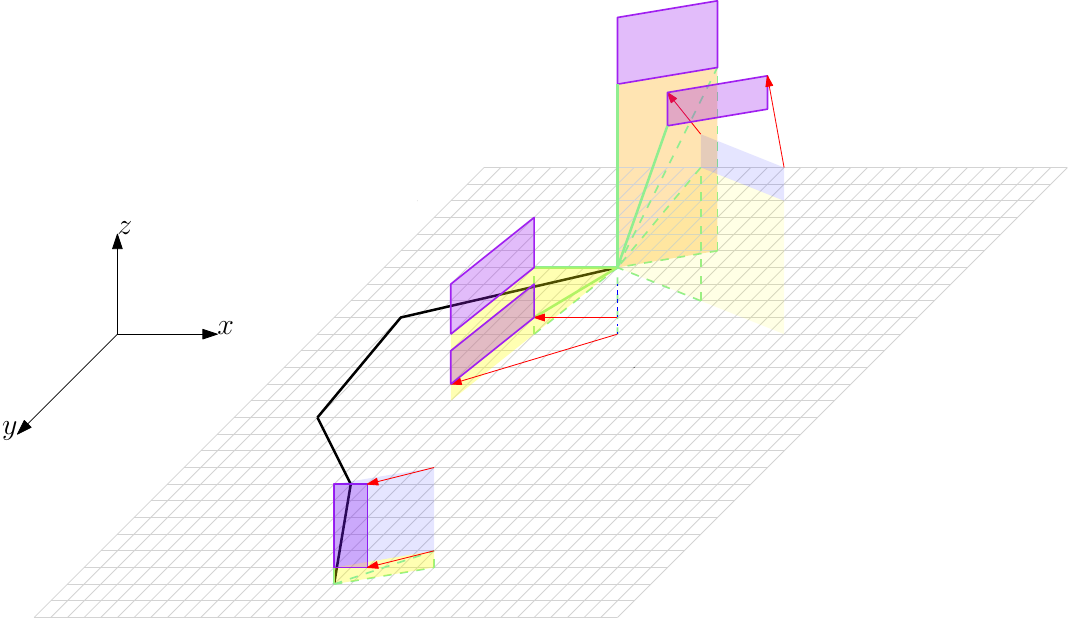}}
\centering
\subfloat[]{\includegraphics[scale = 0.35]{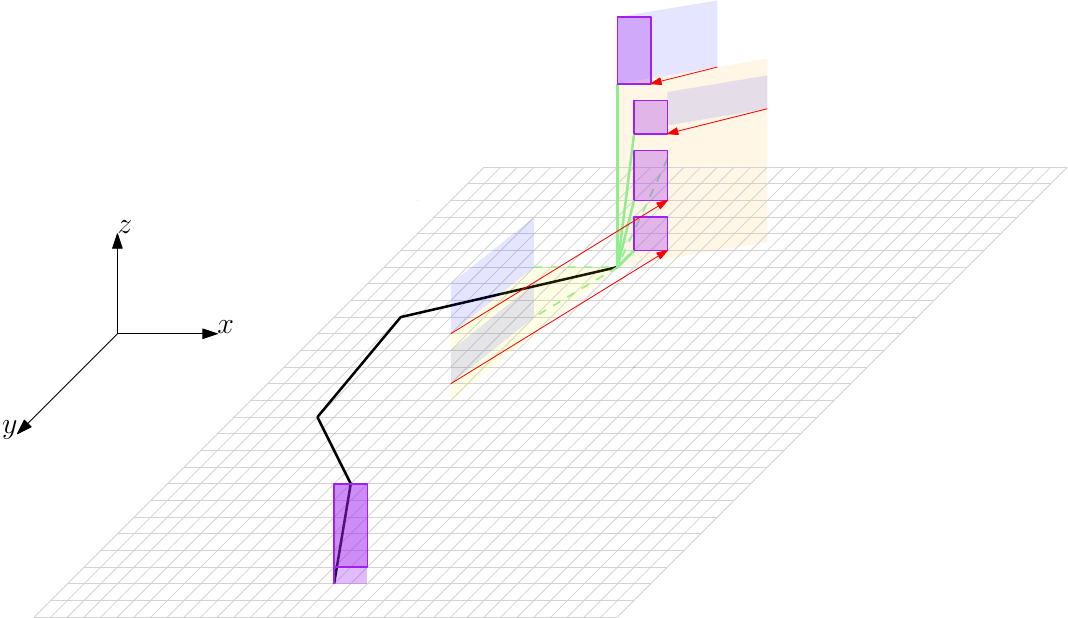}}

  \caption{\textbf{Step 5.} Step 5 in this example consists of two morphing steps.}
\end{figure}

During the following steps $\langle \Gamma_{t + 4}, \Gamma_{t + 5} \rangle, \ldots,\langle \Gamma_{t + 5 + \log k},\Gamma_{t + 5 + \log k + 1} \rangle$ we iteratively divide half-planes that contain $T'(end(e)), e \in K_i$ around each vertex $st(e), e \in K_i$ in pairs, pairs are formed of neighboring half-planes in clockwise order around the pole through  $st(e)$. If in some iteration there are an odd number of planes around some pole, the plane without pair does not move in this iteration. In every iteration we map the drawing of one plane in the pair to another simultaneously in all pairs. As around each vertex we can have at most $k = \Delta(T)$ number of half-planes, we need at most $\mathcal{O}(\log k)$ number of mapping steps to collapse all planes in one and to rotate the resulting image to $XZ^+_{st(e)}$}

\begin{lemma}
Step 5 is a crossing-free morph.
\end{lemma}

\begin{proof}
Since mapping is a crossing-free morph, no intersection can happen in a particular subtree. We now argue that no intersection can happen between two different subtrees.

After step 4, each $T'(end(e)), e \in K_i$ is in the canonical position with respect to $st(e)$ but mapped from $XZ^+_{st(e)}$ to  the vertical half-plane containing $e$. 
During the mapping morph,  we keep the invariant that the closest integral point to the pole in one plane  maps to the closest integral point  in the  other plane.  Then, other integer points are mapped proportionally. 

This implies that  after mapping the plane $\alpha$ containing one subtree to the plane $\beta$ containing another subtree,  in the plane $\beta$, we have two subtrees drawn in such a way that as if  it is  obtained by mapping the canonical drawings of the both the subtrees with respect to $st(e)$ from the $XZ^+_{st(e)}$ to $\beta$. So during each morphing step and at the end of every morphing step no intersections can happen.

Every mapping step for every $st(e)$ is happening in projection to $XY_0$ in its disk \ball{\Gamma_1}{v}{rpw \cdot d(\Gamma) \cdot (4 \cdot d(\Gamma) + 1)} so mappings for different $st(e)$ can not intersect too. 

All other non-processed vertices and $st(e)$ do not change their positions during the morph, hence can not make any intersections. \qed
\end{proof}

We perform \liftk{} for each $K_i \in \mathcal{K}$ till we obtain the canonical drawing of $T$.

\begin{lemma} Conditions (I) and (II) hold after performing \liftk{K_i} for each $1 \le i \le m$.
\end{lemma}

\begin{proof}

\begin{itemize}
	\item[(I)] the drawing of $T'(v)$ in $\Gamma_t$ is the canonical drawing of $T'(v)$ with respect to $v$ for any $v \in V(T)$
	
	\underline{Base of the induction: $i = 1$:}
	
	Note that after performing the stretching step the whole $T$ lies on the $XY_0$ plane. Since none of the paths are lifted, condition (I) trivially holds. 
	
	\underline{The induction step:}	
	
	By induction hypothesis every $T'(v)$ for end vertex of any edge in $K_i$  lies in position needed before \liftk{K_i} procedure. After Step 2 all end vertices of the edges of $K_i$ have the canonical $z$-coordinates with respect to the start vertices of their edges. Then Step 4 for any edge $e$ in $K_i$ makes distances from the vertex $end(e)$ and vertices of its subtree to the pole through $st(e)$ proportional to their canonical $y$-coordinate. After Step 5 all $end(e), e \in K_i$ vertices and their lifted subtrees lie in the corresponding $XZ^+_{st(e)}$ and have the canonical $x, y, z$-coordinates with respect to the start vertices of their edges.
	
	Note that start vertices of edges in $K_i$ did not have lifted subtree in the beginning of the procedure \liftk{K_i} and after this procedure their lifted subtree consists of the end vertices of the edges in $K_i$ and their lifted subtrees.
	
	As for the other vertices in the $XY_0$, their lifted subtrees had not moved during \liftk{K_i} procedure and are in the canonical positions with respect to their roots by induction hypothesis. 
	
	\item[(II)] vertices that lie on edges in a set that is not yet processed are lying within the $XY_0$ plane
	
	\underline{Base of the induction: $i = 1$:}	
	
	The condition (II) holds since the entire tree $T$ lies in the $XY_0$ plane.	
	
	\underline{The induction step:}
	
	By induction hypothesis before \liftk{K_i} procedure all vertices that lie on edges in a set that is not yet processed lie in the $XY_0$ plane. End vertices of the edges in $K_i$ can not lie in $K_j, j > i$ by definition of the partition $\mathcal{K}$. That means that after \liftk{K_i} in which we move only end vertices of the edges in $K_i$ and their lifted subtrees, all vertices that lie on edges of sets $K_j, j > i$ will still lie in the $XY_0$ plane.
\end{itemize}  \qed
\end{proof}

\begin{lemma}
All morphing steps in the algorithm are integer.
\end{lemma}

\begin{proof}
Let us prove this by induction on number of lifted sets $K_i$.

The base case is trivial. After the stretching morph all coordinates of all vertices are integer because the constant of stretching is integer and in the given drawing $\Gamma$ all vertices had integer coordinates.

By induction hypothesis in the beginning of \liftk{K_i} procedure all vertices lie on lattice points of the grid. During \liftk{K_i} procedure vertices that are not ends of the edges of $K_i$ do not change coordinates at all. Mapping morphs move points with integer coordinates to points with integer coordinates.  In all other steps the lifted subtrees of end vertices of the edges are moved along with their roots by the integer vector.

As all coordinates at the beginning of \liftk{K_i} were integer and all vectors of movement of all vertices in every step were integer, after \liftk{K_i} procedure all vertices still have integer coordinates. \qed
\end{proof}

\paragraph{\textbf{Complexity of the algorithm}}

The estimation on the size of the required grid is similar to the previous algorithm:

In the beginning of the algorithm the given  drawing $\Gamma = \Gamma_0$ of graph $T = (V, E)$, $|V| = n$ takes space $l(\Gamma) \times w(\Gamma) \times 1$. First step of the algorithm multiplies needed space by $\mathcal{S}_1 = 2 \cdot rpw \cdot d(\Gamma) \cdot (4 \cdot d(\Gamma) + 1)$.  During our algorithm the height of lifted subtrees does not exceed their height in relative canonical drawing with respect to their roots, i.e. does not exceed $n$. Lifted subtrees take no more than   $rpw \cdot d(\Gamma) \cdot (4 \cdot d(\Gamma) + 1)$ space
in $x; y$-directions from their root during each iteration of \liftk{} procedure. Since in Step 1 we are moving a lifted subtree rooted at $end(e)$ toward $st(e)$ along the edge $e$, the  maximum of the $x; y$-coordinates of vertices in the lifted subtree can not exceed maximum of the $x; y$-coordinates of the vertices $end(e), st(e)$.  During Steps 3-5 of \liftk{} procedure lifted subtrees lie in disks $\ball{\Gamma_1}{v}{rpw \cdot d(\Gamma) \cdot (4 \cdot d(\Gamma) + 1)}$ for some $v$ that is the start of edge in $K_i$.

So the space algorithm takes during all step is at most:
 $$(x \times y \times z): \; \mathcal{O}((l(\Gamma) \cdot 2 \cdot rpw \cdot d(\Gamma) \cdot (4 \cdot d(\Gamma) + 1) + 2 \cdot rpw \cdot d(\Gamma) \cdot (4 \cdot d(\Gamma) + 1)) \times $$
 $$ \times (w(\Gamma) \cdot 2 \cdot rpw \cdot d(\Gamma) \cdot (4 \cdot d(\Gamma) + 1) + 2 \cdot rpw \cdot d(\Gamma) \cdot (4 \cdot d(\Gamma) + 1)) \times n) = $$
 $$ \mathcal{O}(d^3(\Gamma) \cdot \log n \times d^3(\Gamma) \cdot \log n \times n)$$
 
There are $l$ number of sets $K_i$, where $l$ is the depth of $T$. For every procedure \liftk{K_i} we need at most $6 + \log k$ number of morphing steps, where $k = \Delta(T)$, i.e., maximum degree of a vertex in $T$. This implies that the total number of steps in the algorithm is $\mathcal{O}(dpt(T) \cdot \log \Delta(T))$.

The lemmas proved in this section along with the space and time complexity bounds prove the following Theorem.

\begin{theorem} 
\label{theorem:2alg}
For every two planar straight-line grid drawings $\Gamma, \Gamma'$ of a $n$-vertex tree $T$,  there exists a crossing-free 3D-morph $\mathcal{M} = \langle \Gamma = \Gamma_0, \ldots, \Gamma_k = \Gamma' \rangle$ that takes $\mathcal{O}(dpt(T) \cdot \log \Delta(T))$ steps and $\mathcal{O}(d^3 \cdot \log n \times d^3 \cdot \log n \times n)$ space such that every intermediate drawing $\Gamma_i, 0 \le i \le k$ is a straight-line 3D grid drawing, where $d$ is maximum of the diameters of the given drawings, $\Delta(T)$ is the maximum degree of vertices in $T$.
\end{theorem}

\begin{corollary} 
For every two planar straight-line grid drawings $\Gamma, \Gamma'$ of a $n$-vertex tree $T$,  there exists a crossing-free 3D-morph $\mathcal{M} = \langle \Gamma = \Gamma_0, \ldots, \Gamma_k = \Gamma' \rangle$ that takes $\mathcal{O}(dpt(T) \cdot \log n)$ steps and $\mathcal{O}(d^3 \cdot \log n \times d^3 \cdot \log n \times n)$ space such that every intermediate drawing $\Gamma_i, 0 \le i \le k$ is a straight-line 3D grid drawing, where $d$ is maximum of the diameters of the given drawings.
\end{corollary}

\section{Trade-off}\label{Sec:6}

Recall that $\mathcal{L}(T)$ is the set of paths induced by the long-path decomposition, see Section~\ref{sec:defs}. Let $Long(T)$ be a set of paths from $\mathcal{L}(T)$, consisting of the paths which length is at least $\sqrt{n}$, i.e. $Long(T) = \{L_i \in \mathcal{L}(T): |L_i| \ge \sqrt{n}\}$, let the order in $Long(T)$ be induced from the order in $\mathcal{L}(T)$. We denote by $Short(T)$ a set of trees that are left after deleting from $T$ edges of $Long(T)$.

\begin{lemma}
\label{lem:sec_6}
\begin{enumerate}
	\item $|Long(T)| \leq \sqrt{n}$
	\item For every tree $T_i$ in $Short(T)$ depth of $T_i$ is at most $\lfloor \sqrt{n} \rfloor$. 
\end{enumerate}
\end{lemma}
 
\begin{proof}
\begin{enumerate}
	\item Every edge in the tree lies in exactly one path of long-path decomposition. In a tree $T$ with $n$ nodes there are $n - 1$ edge. 
	
	$$n - 1 \geq |E(T)| = |\cup Long(T)| \geq |Long(T)| \cdot \sqrt{n}$$
	
	$$\Rightarrow |Long(T)| \leq \sqrt{n}$$
	\item If there exists a tree $T_i$ which depth is at least $\sqrt{n}$, then long edges from its root create a path that lies in long-path decomposition and has length at least $\sqrt{n}$. That means that this path from the root of $T_i$ should lie in $Long(T)$ and does not exist in $T_i$. We have come to a  contradiction. \qed
\end{enumerate} 
\end{proof}

We divide edges in $Short(T)$ into disjoint sets $Sh_1, \ldots Sh_{\lfloor \sqrt{n} \rfloor}$. An edge $(v_i, v_j)$ in tree $T_k$ lies in the set $Sh_l$ if and only if $max(dpt(v_i), dpt(v_j)) = \lfloor \sqrt{n} \rfloor - l + 1$, where $dpt(v)$ is the depth of vertex $v$ in the corresponding tree $T_k$, see Figure~\ref{fig:trade-off}. Since the maximum depth of any trees $T_k$ is at most $\sqrt{n}$, $Sh_1, \ldots Sh_{\lfloor \sqrt{n} \rfloor}$  contains all the edges of these subtrees.\\
\begin{figure}[!htp]

  \begin{center}  
  \includegraphics[scale = 0.5]{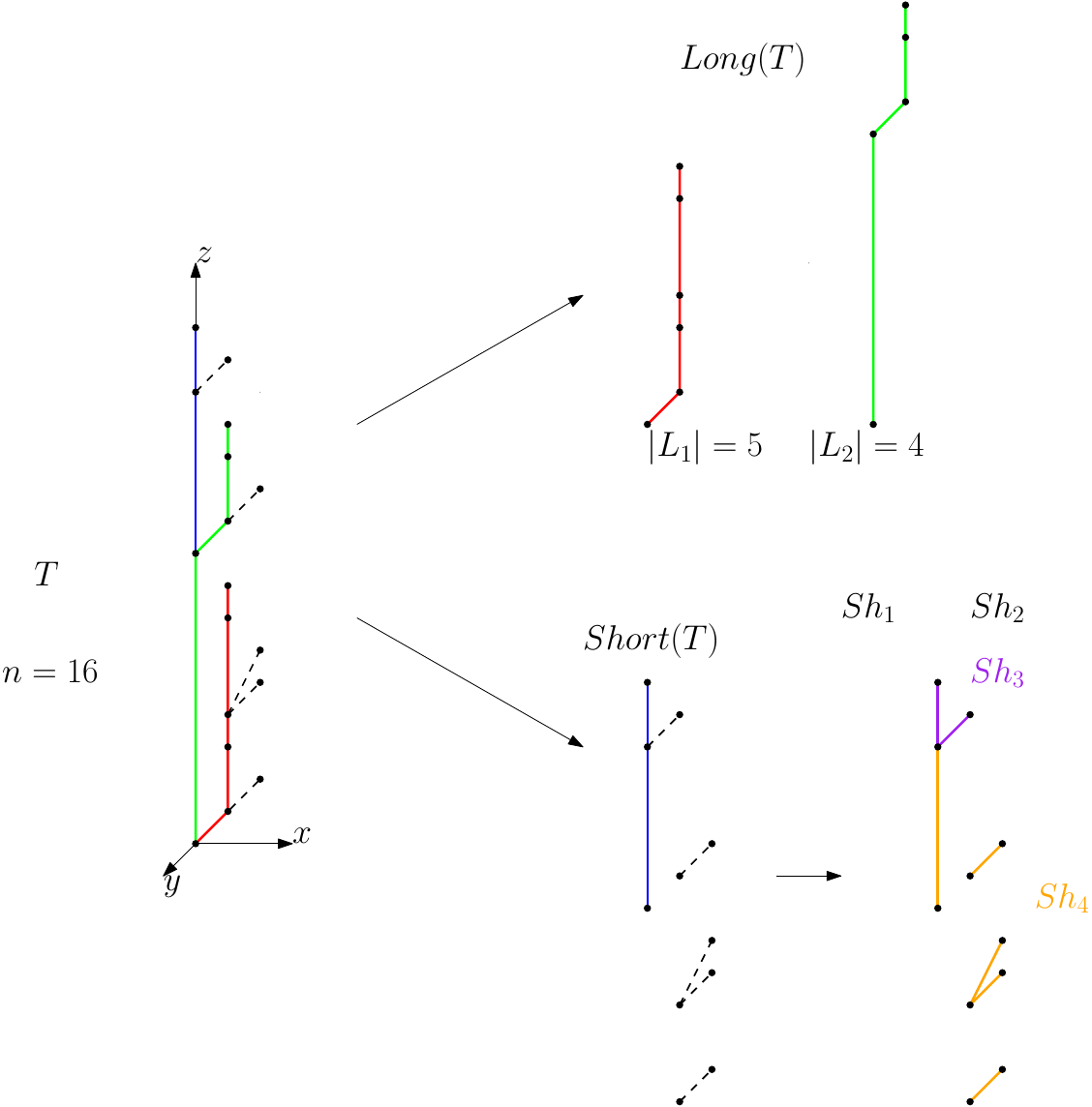}
  \end{center}

  \caption{Example of partition of the edges into set of paths $Long(T)$ and sets of edges $Sh_i, i \in \{1, \ldots, \sqrt{n} = 4\}$. Sets $Sh_1, Sh_2$ are empty, as trees in $Short(T)$ have depth at most 2.}
  \label{fig:trade-off}
\end{figure}

\noindent\textit{\textbf{Trade-off algorithm}:}In the beginning we perform a stretching step with $\mathcal{S}_1 = 2 \cdot rpw \cdot d(\Gamma) \cdot (4 \cdot d(\Gamma) + 1)$ as mentioned in Section~\ref{sec:second_alg}.  $\mathcal{S}_1$ is big enough to perform \lift{} procedure mentioned in Section~\ref{Sec:4}. Then, we lift edges from sets $Sh_1$ to $Sh_{\lfloor \sqrt{n} \rfloor}$ by \liftk{Sh_i} procedure.  This step takes $\mathcal{O}(\sqrt{n} \cdot \log \Delta(T))$ steps in total by Theorem~\ref{theorem:2alg}. After that, we lift paths in $Long(T)$ in reverse order. As $|Long(T)| \leq \sqrt{n}$ and each \lift{} procedure consists of a constant number of morphing steps, this step takes $\mathcal{O}(\sqrt{n})$ steps.

\begin{theorem}
For every two planar straight-line grid drawings $\Gamma, \Gamma'$ of tree $T$ with $n$ vertices there exists a crossing-free 3D-morph $\mathcal{M} = \langle \Gamma = \Gamma_0, \ldots, \Gamma_l = \Gamma' \rangle$ that takes $\mathcal{O}(\sqrt{n} \cdot \log \Delta(T))$ steps and $\mathcal{O}(d^3 \cdot \log n \times d^3  \cdot \log n \times n)$ space to perform, where $d$ is maximum of the diameters of the given drawings, $\Delta(T)$ is the maximum degree of vertices in $T$. In this morph, every intermediate drawing $\Gamma_i, 1 \le i \le l$ is a straight-line 3D grid drawing.
\end{theorem}

Note that it is possible to morph between $\Gamma, \Gamma'$ using $\mathcal{O}(\sqrt{n})$ steps if the maximum degree of $T$ is a constant.

\section{Conclusion}\label{Sec:7}

In this paper, we presented an algorithm that morphs between two planar grid drawings of a $n$-vertex tree $T$ in $\mathcal{O}(\sqrt n \log n)$ steps such that all intermediate drawings are crossing free $3D$-grid drawings and lie inside a polynomially bounded $3D$-grid.  Arseneva et al.~\cite{PoleDance} proved that
$\mathcal{O}(\log n)$ steps are enough to morph between two planar grid drawings of a $n$-vertex tree $T$ where intermediate drawings are allowed to lie in $\mathbb{R}^3$ but they did not guarantee that intermediate drawings have polynomially bounded resolution.  Several problems are left open in this area of research. We mention few of them here.
It is interesting to prove a lower-bound on the number of morphing steps if  intermediate drawings are allowed to lie in $\mathbb{R}^3$. It is also interesting to prove a lower bound on this problem along with the additional constraint of polynomially bounded resolution. Is it possible to morph between two planar grid drawings in $o(n)$ number of steps for a richer class of graphs (e.g. outer-planar graphs) than trees if we are allowed to use the third dimension? Is there a trade-off between the number of steps required and the volume of the grid needed for the morph?

\section*{Acknowledgement}
Elena Arseneva was partially supported by the Foundation for the Advancement of Theoretical Physics and Mathematics “BASIS”. Elena Arseneva and Aleksandra Istomina were partially supported by RFBR, project 20-01-00488. Rahul Gamgopadhyay was   supported by Ministry of Science and Higher Education of the Russian Federation, agreement no. 075–15–2019–1619.

\bibliographystyle{plainurl}
\bibliography{morphing}




\end{document}